\documentclass[a4paper,11pt]{article}

\usepackage{authblk}
\usepackage{amsmath}
\usepackage{amsfonts}
\usepackage{amsthm}
\usepackage{color}
\usepackage{dsfont}
\usepackage{a4wide}
\usepackage{graphicx}
\usepackage{hyperref}

\newcommand{\ui}{\mathrm{i}}
\newcommand{\ue}{\mathrm{e}}
\newcommand{\Id}{\mathrm{Id}}
\newcommand{\Sp}{\mathrm{Sp}}
\newcommand{\grad}{\mathrm{grad}}
\newcommand{\ham}{{\mathcal H}}
 
\newcommand{\veps}{\varepsilon}
\newcommand{\Fn}{F_{\rm n}}

\newcommand{\R}{\mathds{R}}

\newcommand{\N}{\mathds{N}}
\newcommand{\C}{\mathds{C}}

\newcommand{\cA}{{\mathcal A}}

\providecommand{\norm}[1]{\lVert#1\rVert}
\providecommand{\abs}[1]{\lvert#1\rvert}

\newcommand{\spann}{\operatorname{span}}
\newcommand{\tr}{\operatorname{tr}}
\newcommand{\pa}{\partial}
\newcommand{\la}{\langle}
\newcommand{\ra}{\rangle} 

\newcommand{\CE}{{\cal E}}
\newcommand{\CF}{{\cal F}}

\newcommand{\bx}{\mathbf{x}}
\newcommand{\by}{\mathbf{y}}
\newcommand{\bc}{\mathbf{c}}

\newcommand{\Op}{\operatorname{Op}}
\renewcommand{\Im}{\operatorname{Im}}
\renewcommand{\Re}{\operatorname{Re}}

\newcommand{\im}{\operatorname{range}}

\newtheorem{thm}{Theorem}[section]
\newtheorem{lem}[thm]{Lemma}
\newtheorem{prop}[thm]{Proposition}
\newtheorem{cor}[thm]{Corollary}
\newtheorem{Def}[thm]{Definition}

\newtheorem{rem}[thm]{Remark}
\title{Non-Hermitian propagation of Hagedorn wavepackets}

\author[*]{Caroline Lasser}
\affil[*]{Zentrum Mathematik, Technische Universit\"at M\"unchen, 80290 M\"unchen, Germany}
\author[**]{Roman Schubert}
\affil[**]{School of Mathematics, University of Bristol, Bristol, BS8 1TW, UK}
\author[*]{Stephanie Troppmann}

\begin{document}
\maketitle

\begin{abstract}
We investigate the time evolution of Hagedorn wavepackets by non-Hermitian quadratic Hamiltonians. We state a direct connection between coherent states and Lagrangian frames. For the time evolution a multivariate polynomial recursion is derived that describes 
the activation of lower lying excited states, a phenomenon unprecedented for Hermitian propagation. Finally we apply the propagation of excited states to the Davies--Swanson oscillator. 
\end{abstract}



\numberwithin{equation}{section}

\section{Introduction}

In the last two decades considerable interest in non-selfadjoint operators has developed, and even the simplest examples have provided phenomena that greatly 
differ from what is established in the familiar Hermitian context. This pronounced deviation is at the core of the theory of  quantum physical 
resonances \cite{Moi11} and has strongly motivated the research on the pseudospectrum of non-selfadjoint operators \cite{TE05}. 
Our investigation of non-Hermiticity will concentrate on the initial value problem
\begin{equation}\label{eq:ivp}
\ui\veps\partial_t \psi(t) = \Op[\ham_t]\psi(t) \ ,\qquad \psi(0) = \psi_0 \ ,
\end{equation}
where $\veps>0$ is a fixed positive parameter, $\psi_0\in L^2(\R^n)$, and $\Op[\ham_t]$ is the Weyl quantised operator of the quadratic function
$$
\ham_t(z) = \tfrac12 z\cdot H_t z \ ,\qquad z\in\R^{2n} \ ,
$$
associated with a possibly time-dependent complex symmetric matrix $H_t\in\C^{2n\times 2n}$. This seemingly simple model problem 
already encorporates several non-Hermitian challenges and clearly hints at the behaviour of more general systems in the semiclassical limit $\veps\to 0$.

So far, non-Hermitian harmonic systems have been mostly analysed from the spectral point of view or in the context of $\mathcal PT$ symmetry, 
see for example \cite[\S 3]{Sjo74},\cite{Dav99} or more recently \cite{CGHS12,KSTV15}. 
It has been proven that the condition number of the eigenvalues of non-Hermitian harmonic systems grows rapidly with respect to their size \cite{DavKui04,Hen14}, 
while spectral asymptotics have been obtained for skew-symmetric perturbations of harmonic oscillators as well as for non-selfadjoint system with double characteristics \cite{GGN09,HP13,HSV13}. The semigroup of non-selfadjoint quadratic operators has been analysed in \cite{Pra08} and \cite{AV15,Vio16}.
A complementary line of research \cite{GrSch11,GrSch12} has emphasised the new, unexpected geometrical structures that emerge for the non-Hermitian propagation of Gaussian coherent states. Our aim here is to extend these geometrical findings to the larger class of Hagedorn wavepackets and to add the explicit description of additional non-Hermitian signatures of the dynamics.

Over decades, the Hermitian time evolution of Gaussian and Hagedorn wavepackets has evolved into a very versatile tool with wide-ranging application in many areas. 
It was realised by Hepp and Heller \cite{Hep74, Hel75, Hel76} that in order to compute the propagation in the semiclassical limit one only needs one classical trajectory $t\mapsto z_t$ through the centre of the wavepacket and the linearisation of the classical flow around it, 
$$
\dot S_t = \Omega D^2\ham_t(z_t) S_t \ ,\qquad S_t \in\C^{2n\times 2n}\ ,\qquad \Omega = \begin{pmatrix}0 & -\Id_n\\ \Id_n & 0\end{pmatrix}\ , 
$$
where $D^2\ham_t(z_t)$ denotes the Hessian matrix of $\ham_t$ evaluated in $z_t$. 
This method was widely used in applications, in particular in chemistry \cite{Lit86,YeaUze00}. More recently 
Gaussian wavepackets have been more systematically used as a tool for the numerical analysis of highly oscillatory initial value problems, either within the wide framework of Gaussian beams methods, see for example \cite{HRT13}, or for Hermitian quantum dynamics with Hagedorn wavepackets \cite{Lub08,GH14}. Let us give a brief overview of the concepts we develop in this work.

A Gaussian coherent state is parametrised by a phase space point $z_0\in\R^{2n}$ and a normalised Lagrangian frame, that is a rectangular matrix $Z_0\in\C^{2n\times n}$ satisfying the conditions 
\begin{equation}\label{eq:cond}
Z_0^T\Omega Z_0 = 0\quad\text{and}\quad \tfrac{1}{2\ui}\, Z_0^*\Omega Z_0 = \Id_n \ .
\end{equation}
Assuming $z_0=0$ for the ease of notation, we introduce the associated lowering operator
$$
{A(Z_0) = \frac{\ui}{\sqrt{2\veps}} Z_0^T\Omega\hat z} \ ,  \qquad \hat z = \begin{pmatrix}-\ui\veps\nabla_x\\ x\end{pmatrix} \ ,
$$
and the raising operator 
\[
A^{\dagger}(Z_0) = -A(\bar Z_0)
\]
as its formal adjoint, where $\bar Z_0$ denotes the complex conjugate of the matrix $Z_0$. The wavepacket $\varphi_0(Z_0)$ is defined as an element in the kernel 
 of the lowering operator~$A(Z_0)$, i.e. by $A(Z_0)\varphi_0(Z_0)=0$, and it is normalised according to
$$
\|\varphi_0(Z_0)\|^2 =  \int_{\R^n} |\varphi_0(Z_0)(x)|^2 dx = 1.
$$
This determines $\varphi_0(Z_0)$ up to a phase factor.  The Gaussian wavepacket $\varphi_0(Z_0)$ is the zeroth element of an orthonormal basis of $L^2(\R^n)$ constructed by the repeated application of the components of the raising operator to the coherent state. The definition
$$
\varphi_\alpha(Z_0) = \frac{1}{\sqrt{\alpha!}} A^\dagger(Z)^\alpha \varphi_0(Z_0) \ ,\qquad \alpha\in\N_0^n \ ,
$$
is due to Hagedorn \cite{Hag98}, who called the basis elements semiclassical wavepackets, simplifying his earlier construction that was based on a less transparent polynomial recursion \cite{Hag85}. The generalised coherent states \cite[\S4.1]{CR12}, that are obtained 
by applying a unitary squeezing transformation to the $n$-fold product of univariate harmonic oscillator eigenfunctions, only differ by a 
phase factor from the semiclassical wavepackets. For our study of non-Hermitian dynamics we follow Hagedorn's 
ladder approach and use the time evolution of both the raising and lowering operators to explicitly describe the propagation of the basis functions.

We concern the propagation of an excited initial wavepacket 
\[
\psi(0) = \varphi_\alpha(Z_0)\ ,\qquad\alpha\in\N_0^n\ .
\] 
In the Hermitian situation the evolution is simply given by $\psi(t) =  \varphi_\alpha(S_t Z_0)$ where the flow matrix $S_t\in{\rm Sp}(n,\R)$ is real symplectic, and $S_tZ_0$ is a normalised Lagrangian frame for all times $t\in\R$. In the non-Hermitian case, $S_t\in{\rm Sp}(n,\C)$ is a complex symplectic matrix, and we have to expect that
\[
S_t^*\Omega S_t \neq \Omega. 
\]
In particular,  the rectangular matrix $S_t Z_0$ violates the second, normalising, condition of \eqref{eq:cond}, and we have to restrict our analysis to 
time intervals $[0,T[$ such that
\[
\tfrac{1}{2\ui} (S_tZ_0)^*\Omega (S_tZ_0) >0,\qquad t\in[0,T[\ .
\]
We use the Hermitian, positive definite matrix  
\[
N_t := \left(\tfrac{1}{2\ui} (S_tZ_0)^*\Omega (S_tZ_0)\right)^{-1/2}
\] 
to construct a normalised Lagrangian frame $Z_t := S_tZ_0 N_t$ with the same range as $S_t Z_0$.
We then obtain that for an initial coherent state $\psi(0) =\varphi_0(Z_0)$, the corresponding solution of the initial value problem \eqref{eq:ivp} is of the form
$$
\psi(t) = \ue^{\beta_t} \varphi_0(Z_t)
$$
with $\|\varphi_0(Z_t)\| = 1$, where the real-valued gain or loss parameter 
$$
\beta_t = \tfrac14 \int_0^t \tr(G_\tau^{-1}\Im H_\tau) d\tau
$$
is determined by the symplectic metric $G_t = \Omega^T\Re(Z_tZ_t^*)\Omega\in\Sp(n,\R)$ associated with the normalised Lagrangian frame~$Z_t$, see also \cite[\S 3]{GrSch12}. Our main new result states the expansion of the time-evolved wavepacket $\psi(t)$ with respect to the orthonormal basis $\varphi_k(Z_t)$, $k \in\N_0^n$, that is parametrised by the normalised Lagrangian frame~$Z_t$. 

For non-Hermitian dynamics, however, the propagated excited states are utterly different. The two Lagrangian frames $S_tZ_0$ and $\bar S_tZ_0$ do not only lose normalisation but also have different ranges.
Therefore, the dynamics also activate lower lying excited states and we obtain
$$
\psi(t)= \ue^{\beta_t }\sum_{|k|\le|\alpha|} a_k(t) \varphi_k(Z_t)\ ,\qquad t\in[0,T[\ ,
$$ 
with expansion coefficients $a_k(t)\in\C$ for $|k|\le|\alpha|$. These coefficients can be explicitly inferred from Theorem~\ref{thm:main}, that proves 
$$
\psi(t) = \frac{\ue^{\beta_t}}{\sqrt{\alpha!}} \,q_\alpha({N_t} A^\dagger(Z_t)) \varphi_0(Z_t) \ ,
$$
where the multivariate polynomials $q_\alpha$ satisfy the recursion relation
$$
q_0(x) = 1 \ ,\qquad q_{\alpha+e_j}(x) = x_jq_\alpha(x) - e_j\cdot M_t \nabla q_\alpha(x) \ ,\qquad j=1,\ldots,n \ .
$$
The complex symmetric matrix 
$$
{M_t=\tfrac14(S_t \bar Z_0)^TG_t(S_t \bar Z_0)}\in\C^{n\times n}
$$
governing the recursion is determined by the symplectic metric $G_t$ and the complex flow $S_t$. In general, the complex symmetric $M_t$ does not have a specific sparsity pattern so that all the $n$ dimensions are coupled within the polynomial recursion.

We have organised the paper as follows: Section~\ref{sec:lag} develops the symplectic linear algebra of complex Lagrangian subspaces required for 
the parametrisation of the Hagedorn wavepackets. Section~\ref{sec:ladders} constructs coherent states, ladder operators and Hagedorn 
wavepackets parametrised by positive Lagrangian frames. Section~\ref{sec:time} is the core of our manuscript. It analyses the non-Hermitian 
time evolution of Hagedorn wavepackets, and in particular proves our main result Theorem~\ref{thm:main}. Section~\ref{sec:swanson} illustrates our results for the one-dimensional Davies--Swanson oscillator and the heat equation. The four appendices summarise elementary facts on Weyl calculus, present a proof of the Riccati  equation for the 
symplectic metric~$G_t$, and discuss basic properties of the multivariate polynomials $q_\alpha$, $\alpha\in\N_0^n$ and the wavepackets in one and two dimensions.  

\section{Lagrangian subspaces}\label{sec:lag}

We start by discussing some symplectic linear algebra with a focus on complex vector spaces and complex matrices. We endow the real vector space $\R^{2n}$ with the standard symplectic form 
$\R^{2n}\times\R^{2n}\to\R$,  $(x,y)\mapsto x\cdot\Omega y$, 
using the invertible skew-symmetric matrix 
$$
\Omega = \begin{pmatrix}0 & -\Id_n\\ \Id_n & 0\end{pmatrix}\in\R^{2n\times 2n} \ .
$$
Matrices $S\in\R^{2n\times 2n}$ respecting the standard symplectic structure satisfy
$S^T\Omega S = \Omega$ and consequently $S^{-1}=\Omega^T S^T\Omega$. 
They are called {\em symplectic} and constitute the symplectic group $\Sp(n,\R)$, see also \cite[\S I.2]{MS98}. Writing a symplectic matrix as $S=(U,V)$ with $U,V\in\R^{2n\times n}$, the complex rectangular matrix $Z=U-\ui V\in\C^{2n\times n}$ satisfies
\begin{equation}\label{eq:Z}
Z^T\Omega Z = 0 \ ,\qquad Z^*\Omega Z = 2\ui\Id_n \ ,
\end{equation}
where $Z^*=\overline Z^T$ denotes the Hermitian adjoint. We see  from the first property of $Z$ that all vectors $l,l'\in \im Z$ satisfy
$$
l\cdot \Omega l'=0 \ .
$$
Such vectors are  called \emph{skew-orthogonal}, and a subspace $L\subset \C^n\oplus \C^n$ is called {\em isotropic}, if all vectors in $L$ are skew-orthogonal to each other. Moreover, $L$ is called \emph{Lagrangian}, if it is isotropic and has dimension $n$, which is the maximal dimension an isotropic subspace can have (by the non-degeneracy of $\Omega$). From the second property of the matrix $Z$, we see that all vectors $l\in \im Z\setminus\{0\}$ satisfy
$$
\tfrac{\ui}{2} (\Omega \bar l)\cdot l >0 \ .
$$
That is, the quadratic form
$$
h(z,z'):=  \tfrac{\ui}{2} (\Omega \bar z)\cdot z'= \tfrac{\ui}{2} \bar z\cdot \Omega^T z',\qquad z,z'\in\C^n\oplus\C^n \ ,
$$
is positive on $\im Z$. Such a Lagrangian subspace is called {\em positive}. 

\begin{rem}
In the literature, \cite{Lub08}, the choice $Z=U+iV$ is more common, but for our purposes the complex conjugate is more natural, since it matches with Hagedorn's notation, \cite{Hag98}. However, the main results hold true for both definitions.
\end{rem}
\subsection{Lagrangian frames}

Rectangular matrices satisfying conditions \eqref{eq:Z} are convenient tools when working with Lagrangian subspaces. In particular, the normalisation condition 
$Z^*\Omega Z = 2 \ui \Id_n$ will be crucial later on when studying the effects of non-Hermitian dynamics. 

\begin{Def}[Lagrangian frame]
We say that a matrix $Z\in\C^{2n\times n}$ is {\em isotropic}, if 
$$
Z^T\Omega Z = 0 \ ,
$$ 
and it is called {\em normalised}, if 
$$
Z^*\Omega Z = 2 \ui \Id_n \ .
$$
An isotropic matrix of rank $n$ is called a {\em Lagrangian frame}. 
\end{Def}

As indicated before, normalised Lagrangian frames are in one-to-one correspondence with symplectic matrices: Writing $S\in\Sp(n,\R)$ as $S=(U,V)$ with $U,V\in\R^{2n\times n}$, then $Z=U - \ui V$ is isotropic and normalised. Vice versa, if $Z\in\C^{2n\times n}$ is a normalised Lagrangian frame, then $S = (\Re(Z),-\Im(Z))$ is symplectic. For a positive Lagrangian subspace $L\subset\C^n\oplus\C^n$ there are plenty of normalised Lagrangian frames spanning $L$. Denoting by
$$
\Fn(L) = \left\{Z\in\C^{2n\times n};\; \im Z = L,Z^*\Omega Z = 2 \ui \Id_n \right\}
$$
the set of normalised Lagrangian frames spanning $L$, we observe that all its elements are related by unitary transformations. Indeed, since any
$Z_0,Z_1\in \Fn(L)$ have the same range, there exists an invertible matrix $C\in\C^{n\times n}$ so that $Z_1 = Z_0 C$, and the normalisation requires that  $C$ is unitary. 

\subsection{Orthogonal projections}
The complex conjugate $\bar L$ of a positive Lagrangian subspace $L\subset\C^n\oplus\C^n$ is Lagrangian, too, and all vectors $l\in\bar L\setminus\{0\}$ satisfy
$$
h(l,l) = \tfrac{\ui}{2}\bar l\cdot\Omega^T l <0 \ ,
$$
so that $\bar L$ is called a {\em negative} Lagrangian. It is clear that $L\cap \bar L=\{0\}$, because if $l\in L\cap \bar L$, then $l$ is real, and hence 
$h(l,l)=\ui l\cdot \Omega l/2=0$, so that $l=0$. Therefore, 
$$
\C^n\oplus\C^n = L\oplus \bar L \ .
$$ 
This decomposition of $\C^n\oplus\C^n$ is orthogonal in the sense that 
$$
h(l,l')=0\quad\mbox{for all}\quad l\in L, l'\in \bar L \ .
$$

\begin{prop}[Projections]\label{prop:proj} 
Let $L\subset\C^n\oplus\C^n$ be a positive Lagrangian and $Z\in\Fn(L)$. 
Then, 
$$
\pi_L = \tfrac{\ui}{2} ZZ^*\Omega^T \quad\mbox{and}\quad \pi_{\bar L} = -\tfrac{\ui}{2}\bar Z Z^T\Omega^T
$$
are the orthogonal projections onto $L$ and $\bar L$, respectively, that is, 
\begin{itemize}
\item[(i)]
$\pi_L\mid_L = \Id_{2n}$, $\pi_L\mid_{\bar L} = 0$ and $\pi_{\bar L}\mid_{\bar L} = \Id_{2n}$, $\pi_{\bar L}\mid_{L} = 0$ \ ,
\item[(ii)]
$\pi_L^2=\pi_L$ and $\pi_{\bar L}^2=\pi_{\bar L}$ \ ,
\item[(iii)]
$h(\pi_L z,z') = h(z,\pi_L z')$ and $h(\pi_{\bar L}z,z')=h(z,\pi_{\bar L}z')$ for all $z,z'\in\C^n\oplus\C^n$ \ ,
\end{itemize}
\end{prop}

\begin{proof} 
To prove $\pi_L\mid_L=\Id_{2n}$ and $\pi_L\mid_{\bar L}=0$, we observe
$$
\pi_L Z = \tfrac{\ui}{2} ZZ^*\Omega^T Z = Z \ ,\qquad \pi_L\bar Z = \tfrac{\ui}{2} ZZ^*\Omega^T \bar Z = 0 \ .
$$
The other properties of $\pi_L$ and $\pi_{\bar L}$ are also proved by short calculations using that $Z$ is isotropic and normalised.
\end{proof}

\subsection{Siegel half space}
A large set of Lagrangian subspaces can be naturally parametrised by complex symmetric matrices. If the Lagrangian is positive or negative, then we encounter 
complex symmetric matrices with positive or negative definite imaginary part, that is, elements of the upper or lower {\em Siegel half space}. 

\begin{lem}[Siegel half space]\label{lem:siegel}
Assume that $L\subset \C^n\oplus \C^n$ is a Lagrangian subspace so that the projection $\C^n\oplus\C^n\to\C^n$, $(p,q)\mapsto p$
is non-singular on $L$. Then there exists a unique 
symmetric $B\in \C^{n\times n}$ such that 
$$
L=\{(Bq,q)\, ;\, q\in \C^n\} \ .
$$ 
The matrix $B$ can be written as $B=PQ^{-1}$,  where $P,Q\in\C^{n\times n}$ are the components of any Lagrangian frame $Z\in\C^{2n\times n}$ spanning $L$, that is,
$$
Z=\begin{pmatrix}P\\ Q\end{pmatrix} \ ,\qquad \im Z = L \ .
$$ 
Furthermore, $L$ is positive (negative) if and only if $\Im B$ is positive (negative) definite. 
\end{lem}
\begin{proof} 
That the projection of $L$ to $\C^n$ is non-singular means that there is a function $f(q)$ such that $L=\{(f(q),q)\, ;\, q\in\C^n\}$ and since 
$L$ is linear $f$ has to be of the form $f(q)=Bq$ for a uniquely determined  matrix $B\in\C^{n\times n}$. Now let us denote $l_B(q)=(Bq,q)$ for $q\in\C^n$.
Since~$L$ is isotropic, we must have 
$$
0=l_B(q)\cdot\Omega l_B(q')=q\cdot (B-B^T)q'
$$ 
for all $q,q'\in\C^n$, hence $B=B^T$. If $Z=(P;Q)$ and $Z_1=(P_1;Q_1)$ are Lagrangian frames spanning $L$, then there is an invertible matrix 
$C\in\C^{n\times n}$ with $Z_1=ZC$, so that 
$$
P_1Q_1^{-1} = PQ^{-1} = B \ .
$$ 
Furthermore, 
$$
h(l_B(q),l_B(q)) = \tfrac{\ui}{2}\, l_B(q)^*\Omega^T l_B(q) = \tfrac{\ui}{2} q^*  (\overline B-B)q = q^*\Im B q
$$
for all $q\in\C^n$, so that $L$ is positive (negative) if and only if $\Im B$ is positive (negative). 
\end{proof}

\subsection{Metric and complex structure}   
The Hermitian squares of  normalised Lagrangian frames have been useful for writing projections on Lagrangian subspaces. We now examine their real and imaginary 
parts to see more of their geometric information unfolding.

\begin{prop}[Hermitian square]\label{prop:Z}
Let $Z\in\C^{2n\times n}$ be a normalised Lagrangian frame. Then, 
$$
ZZ^* = \Re(ZZ^*) - i\Omega \ ,
$$
where $\Re(ZZ^*)\in\Sp(n,\R)$ is a real symmetric, positive definite, symplectic  matrix. In particular, 
$\Re(ZZ^*)^{-1} = \Omega^T\Re(ZZ^*)\Omega$. 
Moreover, 
$$
\Re(ZZ^*)\Omega Z = \ui Z \ ,\qquad \Re(ZZ^*)\Omega\bar Z = -\ui\bar Z \ ,
$$
so that $(\Re(ZZ^*)\Omega)^2 = -\Id_{2n}$.
\end{prop}

\begin{proof}
Writing $\pi_L+\pi_{\bar L}=\Id_{2n}$ in terms of $Z$, we obtain $-\Im(ZZ^*)\Omega^T = \Id_{2n}$. Hence, $\Im(ZZ^*)=\Omega^T$.
This implies symplecticity of the real part, since
\begin{eqnarray*}
\Re(ZZ^*)^T\Omega\Re(ZZ^*) &=&  \tfrac14 (\bar Z Z^T + Z Z^*)\Omega (\bar Z Z^T + Z Z^*)\\
&=& \tfrac14 ( 2i \, Z Z^*-2i\,\bar Z Z^T) = -\Im(ZZ^*) = \Omega \ .
\end{eqnarray*}
Checking positive definiteness, we see 
$$
z\cdot\Re(ZZ^*)z = \tfrac12 z\cdot (ZZ^*z + \bar Z Z^Tz) = |Z^*z| \ge 0
$$
for all $z\in\R^{2n}$.
If $Z^*z=0$, then $ZZ^*z=0$ and $\Im(ZZ^*)z=0$, which means $z=0$. Finally we compute 
$\ui\Re(ZZ^*)\Omega Z= \tfrac{\ui}{2}(ZZ^*+\bar Z Z^T)\Omega Z = -Z$.
\end{proof}

We have already observed that two normalised Lagrangian frames $Z_0,Z_1\in\Fn(L)$ are related by a unitary matrix $C\in\C^{n\times n}$ with $Z_1=Z_0C$. Therefore the Hermitian squares 
$Z_0Z_0^* = Z_1Z_1^*$
 are the same and can be used for defining two key signatures of the Lagrangian~$L$.

\begin{Def}[Metric \& complex structure]
Let $L\subset\C^n\oplus\C^n$ be a positive Lagrangian subspace and $Z\in\Fn(L)$.
\begin{itemize}
\item[(i)] We call the symmetric, positive definite, symplectic matrix 
$$G = \Omega^T \Re(ZZ^*)\Omega$$
the {\em symplectic metric} of $L$.
\item[(ii)]
We call the symplectic matrix
$$J = -\Omega G$$ with $J^2=-\Id_{2n}$
the {\em complex structure} of $L$.
\end{itemize}
\end{Def}

The complex structure $J$ is a symplectic matrix so that $\Omega J$ is symmetric and positive definite.
Such complex structures are called $\Omega$-compatible. That positive Lagrangian subspaces and 
$\Omega$-compatible complex structures are isomorphic to each other, has been observed and proven in \cite[Lemma 2.3]{GrSch12}. 
The complex structure can also be used for concisely writing the orthogonal projections.

\begin{cor}[Orthogonal projections]\label{cor:proj}
Let $L\subset\C^n\oplus\C^n$ be a positive Lagrangian and $J\in\Sp(n,\R)$ its complex structure. Then the orthogonal projections on $L$ and $\bar L$ can be written as
$$
\pi_L = \tfrac12(\Id_{2n} + \ui J) \ ,\qquad \pi_{\bar L} = \tfrac12(\Id_{2n} - \ui J) \ .
$$
\end{cor}

\begin{proof}
Propositions~\ref{prop:proj} and \ref{prop:Z} yield
$$\pi_L = \tfrac{\ui}{2}ZZ^*\Omega^T = \tfrac{\ui}{2}\Re(ZZ^*)\Omega^T + \tfrac12\Id_{2n} = \tfrac12(\Id_{2n} + \ui J) \ .$$
\end{proof}

We can construct a normalised Lagrangian frame from the eigenvectors of the matrix $G\in\Sp(n,\R)$ representing the symplectic metric. To this end recall the basic structure of the 
spectral decomposition of a positive definite symplectic matrix $G$.

\begin{lem}[Spectrum of symplectic metric] \label{lem:Geigv} Suppose $G\in\Sp(n,\R)$ is symmetric and positive, then there exists a basis $u_1,\ldots,u_n, v_1,\ldots,v_n\in\R^{2n}$ such that 
$$Gu_k=\lambda_k u_k\,\, ,\quad Gv_k=\lambda^{-1}_k v_k\,\, ,$$
where $\lambda_k\geq 1$,  for $k=1, \cdots , n$, and for all $j,k=1, \cdots , n$ we have $u_j\cdot u_k=v_j\cdot v_k=v_j\cdot \Omega u_k=\delta_{jk}$ and $u_j\cdot \Omega u_k=v_j\cdot \Omega v_k=0$. 
\end{lem}

\begin{proof}
This result is in principal well known, see e.g., \cite[Lemma 2.42]{MS98} for a similar statement, but it is hard to locate this exact form of it, so let us indicate the basic idea. Since $G$ is symplectic we have $G\Omega=\Omega G^{-1}$, 
and since $G$ is symmetric there exists a basis of eigenvectors. Now let $u_1$ be an eigenvector with eigenvalue $\lambda_1>0$, then $v_1:=\Omega u_1$ satisfies 
$Gv_1=G\Omega u_1=\Omega G^{-1} u_1=\lambda_1^{-1} \Omega u_1=\lambda_1^{-1}v_1$, and hence is an eigenvector with eigenvalue $\lambda_1^{-1}$. So
we can assume $\lambda_1\geq 1$, and $u_1\cdot v_1=0$ since $G$ is symmetric, and if we normalise $u_1$ as $u_1\cdot u_1=1$ then $v_1\cdot \Omega u_1=1$. 
Let $V_1$ be the span of $u_1,v_1$, then $V_1^{\perp}=V_1^{\Omega}$, where $V_1^{\Omega}:=\{v\in V\, :\, v^T\Omega u=0\, \forall\, u\in V_1\}$, this follows since with 
$v_1=\Omega u_1$ and $u_1=-\Omega v_1$ the conditions $v\cdot u_1=0$ and $v\cdot v_1=$ are equivalent to 
$v\cdot \Omega v_1=0$ and $v\cdot \Omega u_1=0$. Therefore $\R^{2n}=V_1\oplus V_1^{\perp}$ and $V_1^{\perp}$ is symplectic and invariant under $G$, hence we can repeat the previous step in $V_1^{\perp}$ 
and arrive after $k$ steps at a basis with the properties claimed in the lemma. 
\end{proof}

\begin{lem}[Normalised Lagrangian frame]\label{lem:diag}
Let $G\in\Sp(n,\R)$ be symmetric and positive definite. Consider an  eigenbasis $u_1,\ldots,u_n, v_1,\ldots,v_n\in\R^{2n}$ of $G$ as described above in Lemma \ref{lem:Geigv} and denote
$$
l_k := \tfrac{1}{\sqrt{\lambda_k}}u_k - \ui\sqrt{\lambda_k}v_k \ ,\qquad k=1,\ldots,n \ .
$$ 
Then, the matrix $Z\in\C^{2n\times n}$ with column vectors $l_1,\ldots,l_n$ is a normalised Lagrangian frame so that $G = \Omega^T\Re(ZZ^*)\Omega$. 
\end{lem}

\begin{proof}
Using the properties of the basis $u_1,\cdots ,u_n,v_1,\cdots , v_n$ from Lemma \ref{lem:Geigv} we find $l_j\cdot \Omega l_k=0$ and 
$l_j^*\cdot \Omega l_k=2\ui \delta_{jk}$, so $Z$ is a normalised Lagrangian frame. 
Furthermore, again using Lemma \ref{lem:Geigv}, we obtain 
$$
\Re(ZZ^*) = \sum_{k=1}^n \Re(l_kl_k^*) = \sum_{k=1}^n \left( \tfrac{1}{\lambda_k} u_ku_k^T + \lambda_k v_kv_k^T\right) = G^{-1} \, \, .
$$
\end{proof}

\section{Raising and lowering operators}\label{sec:ladders}

Coherent states can be characterised by their lowering operators, or annihilators. These are operators with linear symbols, so 
let us briefly define them and review some of their properties. We will denote 
$$
\hat z = \begin{pmatrix}\hat p\\ \hat q\end{pmatrix} \ ,
$$
where $(\hat p\psi)(x) = -\ui\veps\nabla_x\psi(x)$ is the momentum operator and $(\hat q\psi)(x)=x\psi(x)$ the position operator. 

\begin{Def}[Ladder operators]
Let $l\in \C^n\oplus \C^n$, then we will set 
\begin{equation}
A(l):=\frac{\ui}{\sqrt{2\veps}}\, l\cdot \Omega \hat z \ ,\qquad A^\dagger(l) := -A(\bar l) \ ,
\end{equation}
$A(l)$ is called a {\rm lowering operator}, while $A^\dagger(l)$ is called a {\rm raising operator}. 
\end{Def}
The following properties are important but easy to prove.

\begin{lem}[Commutator relations]\label{lem:lin_op}
We have for all $l,l'\in\C^n\oplus\C^n$
\begin{enumerate}
\item[(i)]
$[A(l),A(l')] = -\frac{\ui}{2} \,  l\cdot \Omega l'$ \ ,
\item[(ii)] 
$[A(l),A^\dagger(l')]=\frac{\ui}{2} \,  l \cdot \Omega \bar l' = h(l',l)$ \ .
\end{enumerate}
$A^\dagger(l)$ is (formally) the adjoint operator of $A(l)$.
\end{lem}

\begin{proof}
We use the phase space gradient $\nabla = \nabla_{p,q}$.
Basic Weyl calculus, see appendix~\ref{app:Weyl}, implies for any symbol $b$
$$
[A(l), \Op[b]]= \ui\veps\Op[\nabla A(l)\cdot\Omega \nabla b] = \sqrt{\veps/2} \,\,  \Op[ l\cdot \nabla b] \ ,
$$
since 
$$
\nabla A(l)\cdot\Omega \nabla b = \frac{-\ui}{\sqrt{2\veps}}\Omega l\cdot\Omega \nabla b  =  \frac{-\ui}{\sqrt{2\veps}} l\cdot\nabla b \ .
$$
Choosing $b = \frac{\ui}{\sqrt{2\veps}} l'\cdot \Omega z$, we have $\nabla b = -\frac{\ui}{\sqrt{2\veps}} \Omega l'$ and this gives us $(i)$. 
Part $(ii)$ follows with choosing $b=-\frac{\ui}{\sqrt{2\veps}} \bar l'\cdot \Omega z$ instead.
\end{proof}

We see in particular from the first property that we can create a set of commuting lowering operators if we choose a set of $l$'s which are skew-orthogonal to each other.  Moreover, a Lagrangian subspace parametrises a maximal family of commuting lowering $A(l)$'s.  Following Hagedorn \cite{Hag98}, we combine them as an operator vector.

\begin{Def}[Ladder vectors]
For an isotropic matrix $Z\in\C^{2n\times n}$ with columns $l_1,\ldots,l_n$ we will denote by $A(Z)$ and $A^\dagger(Z)$, the vectors of annihilation and creation operators, respectively,  
\begin{eqnarray*}
A(Z)&:=&(A(l_1), \cdots ,A(l_n))^T =\frac{\ui}{\sqrt{2\veps}} Z^T\Omega \hat z \ ,\\
A^\dagger(Z)&:=&(A^\dagger(l_1), \cdots ,A^\dagger(l_n))^T =\frac{-\ui}{\sqrt{2\veps}} Z^*\Omega\hat z \ .
\end{eqnarray*}
For any multi-index $\alpha\in\N^n_0$, we set
\begin{eqnarray*}
A_{\alpha}(Z) &:=& A(l_1)^{\alpha_1}A(l_2)^{\alpha_2}\cdots A(l_n)^{\alpha_n} \ ,\\
A^\dagger_{\alpha}(Z) &:=& A^\dagger(l_1)^{\alpha_1}A^\dagger(l_2)^{\alpha_2}\cdots A^\dagger(l_n)^{\alpha_n} \ .
\end{eqnarray*}
\end{Def}

Since all the columns of an isotropic matrix are mutually skew-orthogonal, all the components of the annihilation vector $A(Z)$ commute. The same is true for the creation vector~$A^\dagger(Z)$. Therefore, the operator products $A_\alpha(Z)$ and $A_\alpha^\dagger(Z)$ do not depend on the ordering of their individual factors. 

\begin{rem}[Hagedorn's parametrisation]
The ladder parametrisation coincides with the original one of Hagedorn \cite{Hag98}. Considering matrices $A,B\in\C^{n\times n}$ with $A^TB-B^TA = 0$ and $A^*B+ B^*A=2\Id_n$, he sets
\begin{eqnarray*}
\cA_{\rm Hag}(A,B) &=& \frac{1}{\sqrt{2\veps}}\left(B^T\hat q + \ui A^T\hat p\right) \ ,\\
\cA_{\rm Hag}^\dagger(A,B) &=& \frac{1}{\sqrt{2\veps}}\left(B^*\hat q - \ui A^*\hat p\right) \ .
\end{eqnarray*}
We can write 
$$
\cA_{\rm Hag}(A,B) = \frac{\ui}{\sqrt{2\veps}}\left(-\ui B^T\hat q + A^T\hat p \right) = \frac{\ui}{\sqrt{2\veps}}Z^T\Omega\hat z
\quad\text{with}\quad
Z = \begin{pmatrix}\ui  B\\ A\end{pmatrix} \ ,
$$
and quickly convince ourselves that $Z$ is isotropic and normalised as well.
\end{rem}
\subsection{Coherent states}
Coherent states emerge as a joint eigenfunction with eigenvalue $0$ of a family of commuting operators parametrised by a Lagrangian subspace $L\subset\C^n\oplus\C^n$. We set 
$$
I(L):=\{\varphi\in \mathcal{D}'(\R^n)\, ;\, A(l)\varphi=0 \,\, \text{for all} \,\, l\in L\} \ ,
$$
and observe that 
$$
\varphi\in I(L)\quad\mbox{if and only if}\quad A(Z)\varphi=0
$$
for any Lagrangian frame $Z\in \Fn(L)$. The following characterisation of $I(L)$ is quite standard. For instance one can find a similar statement in \cite[Proposition 5.1]{Hor95}, but let us sketch the proof to elucidate how the Lagrangian property implies the Gaussian form. 

\begin{prop}\label{prop:siegel}
Consider a Lagrangian subspace $L=\{(Bq,q)\, ;\, q\in \C^n\}$ parametrised by a symmetric matrix $B\in\C^{n\times n}$. Then, every element in $ I(L)$ is of the form 
\begin{equation}\label{eq:gaussian_B}
\varphi(x)=c\;\ue^{\frac{\ui}{2\veps}x\cdot Bx} \ ,
\end{equation}
for some constant $c\in \C$. Furthermore, $L$ is positive if and only if  $I(L)\subset L^2(\R^n)$. 
\end{prop}

\begin{proof} 
As in the proof of Lemma~\ref{lem:siegel} we denote $l_{B}(x)=(Bx,x)$ for $x\in\C^n$. Let $l\in L$. Then,  
$$
A(l)\ue^{\frac{\ui}{\veps} x\cdot Bx/2}=\frac{\ui}{\sqrt{2\veps}} l\cdot\Omega l_{ B}(x) \ue^{\frac{\ui}{2\veps} x\cdot Bx}=0
$$
using that $l_{B}(x)\in L$ implies $ l\cdot\Omega l_{B}(x) = 0$. Hence $\ue^{\frac{\ui}{2\veps} x\cdot Bx}\in I(L)$,   
and $\ue^{\frac{\ui}{2\veps} x\cdot Bx}\in L^2(\R^n)$ if and only if $\Im B>0$, which is equivalent to the positivity of $L$. To show uniqueness we use that 
$$
\frac{\ui}{\sqrt{2\veps}}\ue^{\frac{\ui}{2\veps} x\cdot Bx}\; \hat p_i\; \ue^{-\frac{\ui}{2\veps} x\cdot Bx}=A(l_{B}(e_i)) \ .
$$
If $\varphi\in I(L)$, then we find
$$
\pa_i\big(\ue^{-\frac{\ui}{2\veps} x\cdot Bx}\varphi(x)\big)=\ue^{-\frac{\ui}{2\veps} x\cdot Bx} \sqrt{\frac{2}{\veps}}\, A(l_{B}(e_i))\varphi(x)=0
$$
for $i=1,\ldots,n$, therefore $\varphi(x)=c\,\ue^{\frac{\ui}{2\veps} x\cdot Bx}$ for some $c\in\C$. Here we used that $L$ has dimension $n$. 
\end{proof}

Hagedorn's raising and lowering operators \cite{Hag98} originate from his earlier parametrisation of coherent states \cite{Hag85}, which can be conveniently expressed in terms of Lagrangian frames. 

\begin{lem}[Coherent states]\label{lem:par} 
Let $L\subset\C^n\oplus\C^n$ be a positive Lagrangian and consider a Lagrangian frame $Z\in\C^{2n\times n}$ spanning $L$. Define $P,Q\in \C^{n\times n}$ by 
$$
Z=\begin{pmatrix} P \\ Q \end{pmatrix} \ .
$$ 
Then, $Q$ and $P$ are invertible and 
\begin{equation}\label{eq:Z_coh_state}
\varphi_0(Z;x):=(\pi \veps)^{-n/4} (\det Q)^{-1/2}\ue^{\frac{\ui}{2\veps} x\cdot PQ^{-1} x} \in I(L) \ .
\end{equation}
Furthermore, $Z$ is a normalised Lagrangian frame if and only if 
$$
\norm{\varphi_0(Z)}^2=\int_{\R^n}|\varphi_0(Z;x)|^2 dx = 1.
$$
If  $C\in \C^{n\times n}$ is non-degenerate then 
\begin{equation}
\varphi_0(ZC)=(\det C)^{-1/2} \varphi_0(Z) \ .
\end{equation}
\end{lem}

\begin{proof}
Rewriting positivity of the Lagrangian $L$ in terms of $P$ and $Q$ gives
$$
\tfrac{1}{2\ui}Z^*\Omega Z = \tfrac{\ui}{2}(P^*Q-Q^*P)>0 \ .
$$
Hence, $\tfrac{\ui}{2}((Py)^*(Qy)-(Qy)^*(Py))>0$ for all $y\in\C^n$ so that $P$ and $Q$ are invertible. Then Lemma~\ref{lem:siegel} and Proposition~\ref{prop:siegel} imply that the Gaussian wave packet of \eqref{eq:Z_coh_state} is an element of $I(L)$. The normalisation of $Z$ is equivalent to $\tfrac{\ui}{2}(P^*Q-Q^*P)=I$, and 
multiplying from the left by ${Q^*}^{-1}$ and from the right with $Q^{-1}$ gives 
$$
\frac{1}{2\ui} (PQ^{-1}-{Q^*}^{-1}P^*)=(QQ^*)^{-1}
$$
which is the same as 
$$
\Im(PQ^{-1}) = (QQ^*)^{-1}\,\, .
$$ 
 This implies that $\varphi_0(Z)$ is normalised, since then 
 \begin{equation*}
\begin{split}
 \int\abs{\varphi_0(Z;x)}^2\, d x&=(\pi \veps)^{-n/2} \abs{\det Q}^{-1}\int \ue^{-\frac{1}{\veps} x\cdot \Im PQ^{-1} x}\, d x\\
 &=\abs{\det Q}^{-1}(\det  \Im PQ^{-1})^{-1/2}=1\,\, .
 \end{split}
\end{equation*}
The relation between the states $\varphi_0(Z_1)$ with $Z_1=ZC$ and $\varphi_0(Z)$ follows by observing that $P_1=PC$ and $Q_1=QC$, hence $Q_1P^{-1}_1=QP^{-1}$ and $\det Q_1=\det Q\det C$. 
\end{proof}

Notice that \eqref{eq:Z_coh_state} defines $\varphi_0(Z;x)$ only up to a phase, because we have not specified the branch of the square root of $\det Q$. 
In practice it will typically be determined by continuity requirements.

\subsection{Orthonormal basis sets}
Let $L\subset\C^n\oplus\C^n$ be Lagrangian. Applying the operators $A^{\dagger}(l)$ multiple times to an element in $I(L)$ will be used to create a basis. To see the basic idea  assume $L$ is positive and 
$\varphi_0\in I(L)$  has norm one,
$$
\|\varphi_0\|^2 = \langle \varphi_0,\varphi_0\rangle = \int_{\R^n}\overline{\varphi_0(x)} \varphi_0(x) dx = 1\ .
$$
Then we can use the relation 
$$
A(l)A^{\dagger}(l')=[A(l),A^{\dagger}(l')] +A^{\dagger}(l)A(l')=h(l',l)+A^{\dagger}(l)A(l')
$$
to obtain
\[
\begin{split}
\la  A^{\dagger}(l)\varphi_0, A^{\dagger}(l')\varphi_0\ra & =\la \varphi_0 ,A(l)A^{\dagger}(l')\varphi_0\ra\\
&=h(l',l)\la \varphi_0 ,\varphi_0\ra+\la \varphi_0 , A^{\dagger}(l)A(l')\varphi_0\ra\\
&=h(l',l) \ ,
\end{split}
\]
where we have used that $A(l')\varphi_0=0$. So if $h(l',l)=0$, then the states $A^{\dagger}(l)\varphi_0$ and $A^{\dagger}(l')\varphi_0$ will be orthogonal to each other. 
It is easy to check that they are both orthogonal to $\varphi_0$ and that $\|A^\dagger(l)\varphi_0\| = 1$  if $h(l,l)=1$. Iterating this construction yields an orthonormal basis.

\begin{thm}[Orthonormal basis]\label{thm:Hag98}
Let $L\subset\C^n\oplus\C^n$ be a positive Lagrangian subspace and $Z\in\Fn(L)$. 
Then for any  normalized $\varphi_0\in I(L)$ the set 
\begin{equation}
\varphi_{\alpha}(Z):=\frac{1}{\sqrt{\alpha!}} \, A_{\alpha}^{\dagger}(Z)\varphi_0 \ ,\qquad \alpha\in\N^n_0 \ ,
\end{equation}
is an orthonormal basis of $L^2(\R^n)$. 
\end{thm}

This result is due to Hagedorn  \cite[Theorem~3.3]{Hag98}:
The normalisation and orthogonality  follows from commutator arguments similar to  the simple case we discussed. Completeness can be derived from the fact that the functions $\varphi_{\alpha}(Z)$ are the eigenfunctions of the number operator 
$$
N(Z) = A(Z)\cdot A^\dagger(Z) = A(l_1)A^{\dagger}(l_1)+\cdots + A(l_n)A^{\dagger}(l_n) \ ,
$$
which is  selfadjoint and has a complete basis of eigenfunctions due to the following Lemma.


\begin{lem}[Number operator]\label{lem:num} 
Let $L\subset\C^n\oplus\C^n$ be a positive Lagrangian subspace. Let $Z\in\Fn(L)$ and $G\in\Sp(n,\R)$ be the symplectic metric of $L$. Then we can write  $N(Z)=A(Z)\cdot A^\dagger(Z)$ as Weyl operator $N(Z)=\Op[\nu]$ with symbol
$$
\nu(z) = \frac{1}{2\veps}(z\cdot Gz + n\veps) \ ,\qquad z\in\R^{2n} \ .
$$
\end{lem}

\begin{proof}
By Lemma~\ref{lem:quad},
$$
A(Z)\cdot A^\dagger(Z) = \tfrac{1}{2\veps}\left(\hat z\cdot \Omega^T Z Z^*\Omega \hat z + \tfrac{\ui\veps}{2} (-2\ui n)\right) = 
\tfrac{1}{2\veps}\left(\hat z\cdot G\hat z + \veps n\right) \ ,
$$
where we have also used that Proposition~\ref{prop:Z} implies $\Omega^TZ Z^*\Omega = G -\ui\Omega$.
\end{proof}

Returning to the previous remark, by Lemma~\ref{lem:num}, $N(Z)$ is the Weyl quantisation of a positive definite quadratic form. By symplectic classification of quadratic forms, see  \cite[Theorem~21.5.3]{Hor94}, such a form is 
symplectically equivalent to a sum of harmonic oscillators, and using the  quantisation of linear symplectic transformations as metaplectic operators, see \cite[\S2.1.1]{CR12},  $N(Z)$ is therefore unitary equivalent to a sum of standard harmonic oscillators.

The orthogonality and normalisation of the basis functions $\varphi_\alpha(Z)$ depend on the normalisation of the matrix $Z$. 
Let us examine the creation process with parameter matrix $ZC$, where $C$ is non-degenerate. Then, $\varphi_0(ZC)\in I(L)$. However, 
the next creation step provides orthogonality if and only if $C$ is unitary, since 
$$
\langle \varphi_{e_j}(ZC),\varphi_{e_k}(ZC)\rangle = h(ZCe_j,ZCe_k) = e^T_k C^*Ce_j \ .
$$ 
Let us expand $\varphi_\beta(CZ)$, a member of the possibly non-orthogonal function set, with respect to the orthonormal basis $\varphi_\alpha(Z)$, $\alpha\in\N^n_0$.

\begin{thm}[Expansion coefficients] Assume $Z\in\C^{2n\times n}$ is isotropic and normalised, and let $C\in \C^{n\times n}$ be non-degenerate, then for all $\alpha,\beta\in\N^n_0$
\begin{equation}
\la \varphi_{\beta}(ZC), \varphi_{\alpha}(Z)\ra =\frac{\sqrt{\alpha!\beta!}}{(\det \overline{C})^{\frac{1}{2}}}
\sum_{\Lambda\in m(\alpha,\beta)} \frac{C^{\Lambda}}{\Lambda!} \ ,
\end{equation}
where we denote
$$
m(\alpha,\beta):=\left\{ \Lambda=(\lambda_{ij})\in \N_0^{n\times n}\,\, ;\, \, {\textstyle\sum_{i=1}^n\lambda_{ij}=\beta_j\, ,\, \sum_{j=1}^n\lambda_{ij}=\alpha_i} \right\} \ ,
$$  
as well as $\Lambda!:=\prod_{ij}\lambda_{ij}!$ and $C^{\Lambda}:=\prod_{ij} c_{ij}^{\lambda_{ij}}$. 
\end{thm}
\begin{proof}
We have $A^{\dagger}(ZC)=C^*A^{\dagger}(Z)$ by definition of $A^{\dagger}(Z)$. So if we write $\by=A^{\dagger}(ZC)$ and $\bx=A^{\dagger}(Z)$, then we have   
$\by=C^*\bx$ and 
$$
 \varphi_{\beta}(ZC) =\frac{1}{\sqrt{\beta!}} \by^{\beta}\varphi_{0}(ZC)=\frac{1}{\sqrt{\beta!}(\det C)^{1/2}} \by^{\beta}\varphi_{0}(Z) \ . 
$$
To proceed we have  to expand $\by^\beta$ in terms of $\bx$. Since $y_i=\overline \bc_i\cdot \bx$, where $\bc_i$ is the $i$'th column vector of $C$, we obtain
$\by^{\beta}=(\overline \bc_1\cdot \bx)^{\beta_1}\cdots (\overline \bc_n\cdot \bx)^{\beta_n}$, and for the individual terms we  use the multinomial expansion, 
$$
(\overline \bc_j\cdot \bx)^{\beta_j}=\sum_{\abs{\lambda^j}=\beta_j} \frac{\beta_j!}{\lambda^j!}\,\overline \bc_j^{\lambda^j} {\bx}^{\beta^j} \ ,
$$
where the $\lambda^j\in \N_0^n$ are multi-indices and for any vector $\bx=(x_1, \cdots ,x_n)^T\in\C^n$ we set $\bx^{\lambda^j}:=\prod_i x_i^{\lambda_i^j}$. 
Multiplying the terms for different $j$ gives 
$$
\by^{\beta}=\sum_{\abs{\lambda^j}=\beta_j, \, j=1, \ldots ,n} \frac{\beta!}{\lambda^1!\cdots \lambda^n!} \overline \bc_1^{\lambda^1}\cdots \overline \bc_n^{\lambda^n} \bx^{\lambda^1+\cdots +\lambda^n} \ .
$$
Therefore we found 
$$
 \varphi_{\beta}(ZC)=(\det C)^{-1/2}\sum_{\abs{\lambda^j}=\beta_j, \, j=1, \ldots ,n} \frac{\sqrt{\beta!}}{\lambda^1!\cdots \lambda^n!}  \overline \bc_1^{\lambda^1}\cdots \overline \bc_n^{\lambda^n} \sqrt{(\lambda^1+\cdots +\lambda^n)!}\varphi_{\lambda^1+\cdots +\lambda^n}(Z)
$$
and taking the overlap with $\varphi_{\alpha}(Z,z)$ and using orthogonality gives 
$$
\la \varphi_{\beta}(ZC), \varphi_{\alpha}(Z)\ra =
(\det {\overline C})^{-1/2}\sum_{\abs{\lambda^j}=\beta_j, \, j=1, \cdots ,n} \frac{\sqrt{\beta!}}{\lambda^1!\cdots \lambda^n!} \bc_1^{\lambda^1}\cdots \bc_n^{\lambda^n} \sqrt{\alpha!}\, \delta_{\alpha, \lambda^1+\cdots +\lambda^n} \ .
$$
If we introduce the matrix $\Lambda$ with columns $\lambda^1, \ldots ,\lambda^n$, then this formula can be rewritten as in the statement.
\end{proof}

\subsection{Phase space centers}

Sofar we have focused on positive Lagrangian subspaces $L$ and Lagrangian frames $Z\in\Fn(L)$ and have discussed coherent 
states $\varphi_0(Z)$ centered at the phase space origin. Now we extend this framework to formal complex centers $z\in\C^n\oplus\C^n$. This generalisation is motivated by the choice of complex Hamiltonians. To give a physically meaningful interpretation of the associated position and momentum further investigation is needed.

\begin{Def}[Centered ladders] For $l,z\in\C^n\oplus\C^n$ we define the ladder operators
$$
A(l,z) := \frac{\ui}{\sqrt{2\veps}} l\cdot\Omega(\hat z -z) \ ,\qquad A^\dagger(l,z) := -A(\bar l,\bar z) \ .
$$
We note that $A(l,0) = A(l)$ and $A^\dagger(l,0) = A^\dagger(l)$.
\end{Def}

Adding a constant to an operator does not change its commutation properties so that Lemma~\ref{lem:lin_op} also applies to $A(l,z)$ 
with $z\neq 0$, and each of the ladder vectors
$$
A(Z,z) : = \frac{\ui}{\sqrt{2\veps}} Z^T\Omega(\hat z -z) \ ,\qquad A^\dagger(Z,z) := \frac{-\ui}{\sqrt{2\veps}} Z^*\Omega(\hat z -\bar z) \ ,
$$
has commuting components, if $Z\in\C^{2n\times n}$ is an isotropic matrix. One can change the center of ladder operators by conjugating with the (Heisenberg--Weyl) translation operator
$$
T(z) = \exp(-\tfrac{\ui}{\veps} z\cdot\Omega \hat z) \ ,\qquad z=(p,q)\in\C^n\oplus\C^n \ ,
$$ 
that acts as $(T(z)\psi)(x) = \ue^{\frac{\ui}{\veps}p\cdot(x-\frac12 q)}\psi(x-q)$ 
on square integrable functions $\psi\in L^2(\R^n)$, which have a well-defined extension to $\C^n$.  
Indeed, it follows easily from the definition that 
$T(w)\hat z \,T(w)^{-1} = \hat z -w$, which directly yields
\begin{equation}\label{eq:trans}
T(w)A(l,z)T(w)^{-1} = A(l,z+w)
\end{equation}
for all $w,z\in\C^n\oplus\C^n$. Therefore, all the previous results can be translated away from the origin. We have:

\begin{thm}[Orthonormal basis]
Let $L\subset\C^n\oplus\C^n$ be a positive Lagrangian subspace and $Z\in\Fn(L)$. 
Let $z=(p,q)\in\C^n\oplus\C^n$. Then every element in 
$$
I(L,z) := \{\varphi\in{\mathcal D}'(\R^n);\; A(l,z)\varphi = 0\;\text{for all}\; l\in L\} 
$$
is a constant multiple of the normalised coherent state
\begin{equation}\label{eq:coh_cent}
\varphi_0(Z,z;x) := (\pi\veps)^{-n/4} (\det Q)^{-1/2} \ue^{\frac{\ui}{2\veps}(x-q)\cdot PQ^{-1}(x-q) + \frac{\ui}{\veps}p\cdot(x-q)} \ ,
\end{equation}
and the set
$$
\varphi_\alpha(Z,z) := \frac{1}{\sqrt{\alpha!}} A_\alpha^\dagger(Z,z)\varphi_0(Z,z) \ ,\qquad \alpha\in\N^n_0 \ ,
$$
is an orthonormal basis of $L^2(\R^n)$.
\end{thm}

\begin{proof}
We use the translation property \eqref{eq:trans} to lift the results of Proposition~\ref{prop:siegel} and Theorem~\ref{thm:Hag98}, that apply for
$$
I(L,0) = I(L)\quad\text{and}\quad \varphi_\alpha(Z,0) = \varphi_\alpha(Z) \ ,
$$
to general centers $z\neq 0$.
\end{proof}

It turns out that one can always reduce to the case with real center $z$. To understand why let 
us ask which conditions on $z,w\in\C^n\oplus\C^n$ must hold so that  $I(L,z)= I(L,w)$. In terms of the annihilation operators this means that for all $l\in L$ and $\varphi\in I(L,z)$ 
$$
A(l,z)\varphi=A(l,w)\varphi \ .
$$
Since 
$$
A(l,z)-A(l,w)=\frac{\ui}{\sqrt{2\veps}} \, l\cdot \Omega (w-z) \ ,
$$ 
this is equivalent to the condition that 
$l\cdot \Omega (z-w)=0$ for all $l\in L$. So $z-w$ has to be skew orthogonal to $L$, but since $L$ is Lagrangian this means that 
$$
z-w\in L \ .
$$
So any two complex centers whose difference is in $L$ define the same ladder operators, and we just have to find $v\in\C^n\oplus\C^n$ so that 
$w = z + \pi_{L}v$ is real. Now we use Corollary~\ref{cor:proj} and write $\pi_{L}v = \tfrac12(v + \ui Jv)$. Then we immediately see that
$v = -2\ui\Im z$ provides the real center
$$
w = \Re z + J\Im z \ ,\qquad w-z\in L \ .
$$ 
In summary we have reproduced the linear algebra part of \cite[Theorem~2.1]{GrSch12} that also provides the associated coherent states. 

\begin{thm}[Real centers]\label{thm:real}
Let $L\subset\C^n\oplus\C^n$ be a positive Lagrangian and $Z\in\Fn(L)$. Let  $J\in\Sp(n,\R)$ be the complex structure of $L$ and define $P_J: \C^n\oplus\C^n\to\R^n\oplus\R^n$, $P_J(z) = \Re z + J\Im z$. Then, for any $z=(p,q)\in\C^n\oplus\C^n$
$$
A(Z,z) = A(Z,P_J(z)) \ ,\qquad I(L,z) = I(L,P_J(z))
$$
and the coherent states are related by
$$
\varphi_0(Z,z)=\ue^{\frac{\ui}{2\veps } (\eta+p)\cdot (y-q)}\varphi_0(Z,P_J(z)) \ ,\qquad (\eta,y)=P_J(z) \ .
$$
\end{thm}

\section{Time evolution}\label{sec:time}

We will now explore how the  ladder operators, coherent states, and the associated basis behave when we propagate them in time according to a non-Hermitian operator. Let $H_t\in \C^{2n\times 2n}$ be a symmetric matrix,  which depends  continuously on $t$, 
and denote by $\hat \ham_t=\Op[\ham_t]$ the Weyl quantisation of the quadratic function
$$
\ham_t(z) = \tfrac{1}{2} z\cdot H_t z \ ,\qquad z\in\R^{2n} \ .
$$  
We are interested in the time-dependent Schr\"odinger equation
\begin{equation}\label{eq:schro}
\ui\veps\partial_t\psi(t) = \hat \ham_t\psi(t) 
\end{equation}
with initial data $\psi(0)$ that are Hagedorn wavepackets.
If $H_t$ is a real symmetric matrix, then we are in the standard setting. $\hat \ham_t$ is a self-adjoint operator on some dense domain of $L^2(\R^n)$ and defines for all $t\in\R$ a unitary time evolution, see for example \cite[\S3.1]{CR12}. In the time-independent case with $\Im H\le0$, one defines the evolution as a contraction semigroup on $L^2(\R^n)$, see the discussion before \cite[Theorem 4.2]{Hor95} or \cite[Theorem~1]{Pra08}. In the more general case that $\Im H_t$ depends continuously on time and $\Im H_t\leq 0$ the well-posedness of the evolution problem \eqref{eq:schro} in $L^2(\R^n)$ follows from the time-independent case and Kato's results on hyperbolic evolution systems, see \cite[\S5.3]{Paz83}.  
If $\Im H> 0$, then the time evolution might cease to be well-defined after some finite time $T>0$, as shown by the examples in Section~\ref{sec:swanson} or \cite{GrSch12}. However, we will determine time intervals here such that the propagation of Hagedorn wavepackets is well-defined and explore the non-unitary evolution of ladder operators, coherent states and excited states.
We first investigate how a non-vanishing imaginary part changes the geometrical structure. 

\subsection{Metriplectic structure}

We decompose the Hamiltonian function $\ham_t(z)$ into its real and imaginary part and first consider the Schr\"odinger equation for the real part, 
\begin{equation}\label{eq:Schro_real}
\ui\veps\partial_t\psi(t) = \Op[\Re\ham_t]\psi(t) \ .
\end{equation}
Since $\Op[\Re\ham_t]$ is a self-adjoint operator, the conventional Schr\"odinger equation \eqref{eq:Schro_real} can be reformulated as the Hamiltonian equation
$$
\partial_t\psi(t) = X_{\CE_t}(\psi(t))\ ,
$$
where $X_{\CE_t}$ denotes the Hamiltonian vector field
for the energy function
$$
\CE_t(\psi) := \langle \psi,i K \psi\rangle \ ,\qquad\psi\in H^2(\R^n) \ ,
$$
with $K = \tfrac{1}{\ui\veps}\Op[\Re\ham_t]$, see for example \cite[Corollary~2.5.2]{MaRa99}. Indeed, one equips the complex Hilbert space $L^2(\R^n)$ with the symplectic form
$$
\omega(\varphi,\psi) := 2\Im\langle\varphi,\psi\rangle \ ,\qquad\varphi,\psi\in L^2(\R^n) \ ,
$$
and computes for the derivative of the energy

\begin{eqnarray*}
\langle \CE_t'(\psi),\varphi\rangle &=& 
\lim_{h\to0}\tfrac{1}{h}\left(\CE_t(\psi+h\varphi) - \CE_t(\psi)\right)\\
&=&
\langle \psi,iK\varphi\rangle + 
\langle \varphi,iK\psi\rangle \\
&=&
-i \langle K\psi,\varphi\rangle +i \,
\overline{\langle K\psi,\varphi\rangle}
=\omega(K\psi,\varphi)
\end{eqnarray*}
for all $\varphi,\psi\in H^2(\R^n)$, so that 
$$
\tfrac{1}{\ui\veps}\Op[\Re\ham_t]\psi = X_{\CE_t}(\psi) \ .
$$ 
Now let us consider the more general case of non-Hermitian 
time evolution, which is not captured by symplecticity alone but requires additional metric structure.  We set
\begin{equation}\label{eq:metric}
g(\varphi,\psi) := 2\Re\langle\varphi,\psi\rangle \ ,\qquad \varphi,\psi\in L^2(\R^n) \ , 
\end{equation}
and observe that $g$ is a symmetric and positive definite $\R$-bilinear form. This metric defines the gradient flow contribution generated by the imaginary part. Indeed, setting
$$
\CF_t(\psi) := \langle \psi,\tfrac1\veps \Op[\Im\ham_t]\psi\rangle \ ,\qquad\psi\in H^2(\R^n) \ , 
$$ 
we have
$$
\langle \CF_t'(\psi),\varphi\rangle = 
\langle\psi,\tfrac1\veps\Op[\Im \ham_t]\varphi\rangle + \langle\varphi,\tfrac1\veps\Op[\Im\ham_t]\psi\rangle = 
g(\tfrac1\veps\Op[\Im\ham_t]\psi,\varphi)
$$
for all $\varphi,\psi\in H^2(\R^n)$, since $\Op[\Im \ham_t]$ is self-adjoint. In summary, we can rewrite the non-Hermitian Schr\"odinger equation \eqref{eq:schro} as
\begin{equation}\label{eq:metri_schro}
\partial_t\psi(t) = X_{\CE_t}(\psi(t)) + \grad\,\CF_t(\psi(t)) \ ,
\end{equation}
and we note that such an additive combination of Hamiltonian and gradient structure defines a metriplectic system in the sense of \cite[\S15.4.1]{BMR13}, if additional compatibility conditions on the energies $\CE_t(\psi)$ and $\CF_t(\psi)$ are satisfied.

In the following we will see how a similar metriplectic structure emerges as well in the semiclassical limit of the propagation of coherent and excited states.
\subsection{Ladder evolution}

Let $S_t\in\C^{2n\times 2n}$ be the matrix defined as the solution to 
\begin{equation}\label{eq:flow}
\dot S_t =\Omega H_t S_t \ ,\quad S_0=\Id_{2n} \ , 
\end{equation}
for some time-interval $[0,T[$. It is easy to check that $S_t$ is a complex symplectic matrix, i.e.,
$$
S_t^T\Omega S_t=\Omega \ ,
$$
and if $H_t=H$ does not depend on time, then $S_t=\exp(t\Omega H)$ exists for all $t\in\R$. If $H_t$ is a real matrix, then $S_t$ will be a real symplectic matrix. Otherwise, the matrix $S_t$ is complex.  We first examine the dynamics of the ladder operators with initial center at the origin.

\begin{lem}[Ladder evolution]\label{lem:A_transport}
For all $l\in\C^n\oplus\C^n$, the ladder operators satisfy
\begin{align*}
\ui\veps\partial_t A(S_t l) &= [\hat \ham_t,A(S_t l)]\ ,\\
\ui\veps\partial_t A^\dagger(\bar S_t l) &= [\hat \ham_t,A^\dagger (\bar S_t l)]\ .
\end{align*}
\end{lem}
\begin{proof}
We recall that
\[
A(S_t l) = \frac{\ui}{\sqrt{2\veps}}\, S_t l\cdot\Omega\, \widehat z\quad\text{and}\quad 
A^\dagger(\bar S_t l) = -\frac{\ui}{\sqrt{2\veps}}\, S_t \bar l\cdot\Omega\, \widehat z\ .
\]
Since $\partial_t S_t = \Omega H_t S_t$, we obtain
\[
\ui\veps \partial_t A(S_t l) = -\sqrt{\frac{\veps}{2}} S_t l\cdot  H_t\hat z\quad\text{and}\quad
\ui\veps \partial_t A^\dagger(\bar S_t l) = \sqrt{\frac{\veps}{2}} S_t \bar l\cdot  H_t\hat z\ .
\]
By Weyl calculus, see appendix~\ref{app:Weyl}, we furthermore find for the commutators
$$
[\hat \ham_t, A(S_t l)] = - \sqrt{\frac{\veps}{2}}\Op[(H_tz)\cdot\Omega \Omega^T S_t l] = - \sqrt{\frac{\veps}{2}}  S_t l\cdot H_t \hat z
$$
and
\[
[\hat \ham_t, A^\dagger(\bar S_t l)] = \sqrt{\frac{\veps}{2}}  S_t \bar l\cdot H_t \hat z \ .
\]
\end{proof}
The previous Lemma implies that for any isotropic matrix $Z_0\in\C^{2n\times n}$ the lowering and raising operators 
evolve according to
$$
t\mapsto A(S_t Z_0) \quad\text{and}\quad t\mapsto A^\dagger(\bar S_tZ_0) \ ,
$$
and we observe that both matrices $S_t Z_0$ and $\bar S_tZ_0$ inherit isotropy, since $S_t$ and $\bar S_t$ are symplectic. However, even if $Z_0$ is normalised, neither $S_tZ_0$ nor $\bar S_tZ_0$ need to be normalised, since in general 
$$
S_t^*\Omega S_t\neq\Omega \ .
$$ 
Furthermore, the raising operator is no more the adjoint of the lowering operator, 
$$
A^\dagger(\bar S_t Z_0) \neq A^*(S_t Z_0) \ ,
$$
and if $L_0 =\im Z_0$ is the initial Lagrangian subspace, then $A(S_t Z_0)$ and $A^\dagger(\bar S_t Z_0)$ belong to the different Lagrangian subspaces $S_t L_0$ and $\bar S_t L_0$, respectively. Only if $H_t$ is a real matrix, then $S_tZ_0$ stays normalised, while both the raising and the lowering operators are adjoint to each other and belong to the same Lagrangian subspace $S_tL_0$.

\subsection{Coherent state propagation}

Let us next consider the evolution of coherent states on time intervals $[0,T[$ so that 
$$
L_t := S_t L_0
$$
is a positive Lagrangian subspace. Such intervals exist by continuity of $t\mapsto  S_t$, if the initial Lagrangian $L_0$ is positive. If we propagate a 
normalised Lagrangian frame $Z_0\in\Fn(L_0)$ by the complex flow matrix $S_t$, then $S_t Z_0$ is in general not normalised and the associated coherent state $\varphi_0(S_t Z_0)$ is not normalised either. Hence, we look for a normalised replacement of $S_t Z_0$. Since $L_t$ is positive by assumption, the matrix 
$$
N_t := \left(\tfrac{1}{2\ui} (S_tZ_0)^*\Omega (S_tZ_0)\right)^{-1/2}
$$
is  well-defined, in particular Hermitian and positive definite, so that 
$$
Z_t := S_t Z_0 N_t\in\Fn(L_t) \ .
$$
We note that for every normalised Lagrangian frame $W_t\in\Fn(L_t)$ there exists a complex, invertible matrix $C_t\in\C^{n\times n}$ such that $W_t = S_tZ_0 C_t$. 
The Lagrangian frame $Z_t$ is a particular one in the sense that $N_t$ is Hermitian and positive definite. 
This property will allow us to explicity write the time evolved coherent state in terms of a normalised coherent state accompanied by a positive loss or gain factor.

\begin{prop}[Coherent state evolution]\label{prop:coh}
Let $L_0\subset\C^n\oplus\C^n$  and $L_t= S_t L_0$ be positive Lagrangian subspaces for $t\in[0,T[$. Let $G_t\in\Sp(n,\R)$ be the symplectic metric of $L_t$ and consider $Z_t\in\Fn(L_t)$ so that $Z_t= S_tZ_0N_t$ for a Hermitian positive definite matrix $N_t\in\C^{n\times n}$. If the initial state is given by the coherent state $\varphi_0(Z_0)$ from \eqref{eq:Z_coh_state}, then
$$
\varphi_0(t) = \varphi_0(S_t Z_0) = \ue^{\beta_t} \varphi_0(Z_t) \ ,\qquad t\in[0,T[ \ ,
$$
with
$$
\beta_t = \tfrac14 \int_0^t \tr(G_\tau^{-1} \Im H_\tau) d\tau \ .
$$
\end{prop}

\begin{proof}
The proof for real matrices $H_t$ is well known and goes back to Hagedorn \cite{Hag80}. It
can be extended without any changes to the complex case to obtain
$$
\varphi_0(t) = \varphi_0(S_t Z_0) \ ,
$$
see the proof of Proposition~\ref{prop:coh_z} later on. Switching to the normalised Lagrangian frame $Z_t$, Lemma~\ref{lem:par} implies
$$
\varphi_0(S_t Z_0) = \det(N_t)^{1/2} \varphi_0(Z_t) \ .
$$
By Jacobi's determinant formula $\partial_t\det(N_t) = \det(N_t) \tr(\partial_t N_t N_t^{-1})$ we have
$$
\partial_t \det(N_t)^{1/2} =  \tfrac12 \det(N_t)^{1/2} \tr(\partial_t N_t N_t^{-1}) \ .
$$
We now use the Hamiltonian systems
$$
\pa_t S_t = \Omega H_t S_t \ ,\qquad \pa_t S_t^* = S_t^* \bar H_t \Omega^T \ ,
$$
to differentiate the normalisation property $\frac{1}{2\ui}(S_tZ_0N_t)^*\Omega(S_tZ_0N_t) = \Id_n$.
We obtain
\begin{eqnarray*}
0 &=& \partial_t N_t^* N_t^{-*} +\tfrac{\ui}{2}N_t^* (S_tZ_0)^*(H_t - \bar H_t)(S_tZ_0)N_t + 
N_t^{-1}\partial_t N_t\\
&=& \partial_t N_t N_t^{-1} - Z_t^*\Im H_t Z_t + N_t^{-1}\partial_t N_t \ ,
\end{eqnarray*}
and by Proposition~\ref{prop:Z}
$$
\tr(\partial_t N_t N_t^{-1}) = \tfrac12\tr(Z_t^*\Im H_t Z_t) = \tfrac12\tr \left(\Im H_t(G_t^{-1}-\ui\Omega)\right) = \tfrac12\tr(\Im H_t\, G_t^{-1}) \ .
$$
It remains to write $\det(N_t)^{1/2} =:\ue^{\beta_t}$ and to observe that
$$
\partial_t\beta_t = \tfrac14 \tr(G_t^{-1}\Im H_t) \ ,\qquad \beta_0 = 0 \ .
$$
\end{proof}

The real-valued scalar $\beta_t$ can either be determined via the normalising matrix,
\[
\ue^{\beta_t} = \det(N_t)^{1/2}\ ,
\]
or via the imaginary part $\Im H_t$ of the Hamiltonian matrix together with the symplectic metric $G_t\in\Sp(n,\R)$ according to
$$
\beta_t = \tfrac14 \int_0^t \tr(G_\tau^{-1} \Im H_\tau) d\tau.
$$
It describes the norm of the propagated coherent state, 
$$
\|\varphi_0(t)\| = \|\ue^{\beta_t}\varphi_0(Z_t)\| = \ue^{\beta_t} \ .
$$  
The evolution of the symplectic metric $G_t$ and the corresponding complex structure $J_t=-\Omega G_t$ are governed by the following Riccati equations.

\begin{thm}[Riccati equations]\label{thm:ric}
Let $L_0$ and $L_t=S_tL_0$ be positive Lagrangian subspaces. Denote by $G_t,J_t\in\Sp(n,\R)$ the symplectic metric and the complex structure of $L_t$, respectively. Then, 
\begin{eqnarray*}
\dot G_t &=& \Re H_t \Omega G_t - G_t\Omega \Re H_t - \Im H_t - G_t\Omega \Im H_t \Omega G \ ,\\
\dot J_t &=& \Omega\Re H_t J_t - J_t\Omega \Re H_t + \Omega \Im H_t + J_t\Omega \Im H_t J_t \ .
\end{eqnarray*}
\end{thm}

\begin{proof}
The equations of motion for $G_t$ and $J_t$ have been derived in \cite[Theorem~3.3]{GrSch12} using the Siegel half space and rational relations. The appendix~\ref{app:ric} provides an alternative proof based on Lagrangian frames.
\end{proof}

\subsection{Excited state propagation}
Next let us consider the propagation of first order excited states $A^{\dagger}(l)\varphi_0(Z_0)$ for $l\in L_0$. By Lemma \ref{lem:A_transport} and Proposition~\ref{prop:coh} we obtain that
$$
\varphi_{1}(t) := A^{\dagger}(\bar S_t l)\varphi_0(t)
= \ue^{\beta_t} A^{\dagger}(\bar S_tl)\varphi_0(Z_t) 
$$
satisfies
\begin{align*}
\ui\veps \partial_t \varphi_{1}(t) &= [\hat \ham_t,A^\dagger(\bar S_t l)]\varphi_0(t) + A^\dagger(\bar S_t l) \hat\ham_t \varphi_0(t)\\
&= \hat\ham_t \ \varphi_{1}(t)
\end{align*}
subject to the initial condition $\varphi_{1}(0)=A^{\dagger}(l)\varphi_0(Z_0)$.
If $S_t$ is a complex matrix, then $\bar S_t l\notin  L_t = S_t L_0$ so that $A^{\dagger}(\bar S_t l)$ is not a creation operator associated with $L_t$. We therefore use Proposition~\ref{prop:proj} and decompose  
$$
\bar S_t l=\pi_{ L_t}\bar S_t l+\pi_{\bar L_t}\bar S_t l = \pi_{ L_t}\bar S_t l+ \overline{\pi_{L_t}} \bar S_t l \ ,
$$
which leads to 
\begin{equation}\label{eq:dec}
A^{\dagger}(\bar S_tl)=A^{\dagger}(\pi_{ L_t}\bar S_tl)+A^{\dagger}(\overline{\pi_{L_t}}\bar S_t l)  = A^{\dagger}(\pi_{ L_t} \bar S_tl)-A(\pi_{L_t}S_t \bar l ) \ .
\end{equation}
Therefore,
$$
A^{\dagger}(\bar S_tl)\varphi_0(Z_t)=A^{\dagger} (\pi_{  L_t} \bar S_t l)\varphi_0(Z_t) \ , 
$$
since $A(\pi_{L_t}S_t \bar l)\varphi_0(Z_t)=0$. The following Lemma and Theorem extend this line of argument to vectors of excited states.

\begin{lem}[Ladder decomposition]\label{lem:expand_A}
Let $L_0$ and $L_t = S_t L_0$ be positive Lagrangian subspaces and $Z_t\in\Fn(L_t)$. Then, 
$$
A^{\dagger}(\bar S_tZ_0) = C_t^*A^{\dagger}(Z_t)-D_t^TA(Z_t) \ ,
$$
where $C_t=\tfrac{\ui}{2} Z_t^*\Omega^T \bar S_t Z_0$ and $D_t=\tfrac{\ui}{2} Z_t^*\Omega^T S_t \bar Z_0$ 
are the unique matrices in $\C^{n\times n}$ so that 
$$
\bar S_t Z_0 = Z_tC_t + \bar Z_t \bar D_t \ .
$$ 
\end{lem}

\begin{proof}
We apply the decomposition \eqref{eq:dec} to the column vectors $l_1,\ldots,l_n$ of $Z_0$ and obtain 
$$
A^{\dagger}(\bar S_tZ_0)
=A^{\dagger}(\pi_{L_t}\bar S_tZ_0)-A(\pi_{L_t} S_t \bar Z_0) \ .
$$
Now we want to find $C_t$ and $D_t$ such that 
\begin{equation}\label{eq:CD}
\pi_{L_t}\bar S_tZ_0=Z_tC_t\,\, \quad \text{and}\quad \pi_{L_t} S_t \bar Z_0=Z_tD_t
\end{equation}
because then  $A^{\dagger}(\pi_{L_t}\bar S_tZ_0)=C_t^*A^{\dagger}(Z_t)$ and 
$A(\pi_{L_t} S_t \bar Z_0)=D_t^T A(Z_t)$, which will give the result. 
We just multiply the equations in \eqref{eq:CD} from the left by 
$\frac{\ui}{2}Z_t^*\Omega^T$ and use the normalisation of~$Z_t$, which gives 
$$
C_t=\frac{\ui}{2} Z_t^*\Omega^T \pi_{L_t} \bar S_t Z_0\,\, \quad \text{and} \quad D_t=\frac{\ui}{2} Z_t^*\Omega^T \pi_{L_t} S_t \bar Z_0 \ .
$$
Now it remains to compute $\frac{\ui}{2} Z_t^*\Omega^T \pi_{L_t}=\frac{\ui}{2} Z_t^*\Omega^T \frac{\ui}{2} Z_tZ_t^*\Omega^T =\frac{\ui}{2} Z_t^*\Omega^T$, 
and this leads to the claimed expressions for $C_t$ and $D_t$. Adding the defining equations in \eqref{eq:CD}, we finally obtain
$$
Z_t C_t + \bar Z_t\bar D_t = \pi_{L_t}\bar S_t Z_0 + \overline{\pi_{L_t}} \bar S_t Z_0 = \bar S_t Z_0 \ . 
$$
\end{proof}

Since $A(Z_t)\varphi_0(Z_t) = 0$, the Lemma implies 
\begin{equation}\label{eq:first}
\ue^{\beta_t} A^{\dagger}(\bar S_tZ_0)\varphi_0(Z_t) = \ue^{\beta_t}\, C_t^*A^{\dagger} (Z_t)\varphi_0(Z_t) \ .
\end{equation}
Let us next consider more highly excited states and expand $\varphi_\gamma(t)$, $|\gamma|>1$, with
$$
\varphi_\gamma(0)=\varphi_\gamma(Z_0)= \frac{1}{\sqrt{k!}}A^\dagger_\gamma(Z_0)\varphi_0(Z_0)
$$ 
in terms of the orthonormal basis $\varphi_\alpha(Z_t)$, $\alpha\in\N_0^n$. Now the annihilation part of the decomposition
$$
A^{\dagger}(\bar S_tZ_0) = C_t^*A^{\dagger}(Z_t)-D_t^TA(Z_t)
$$
becomes more visible and we encounter commutators between $C_t^*A^{\dagger}(Z_t)$ and $D_t^TA(Z_t)$. In this situation, the term handling is facilitated by multivariate polynomial recursions that are governed by a complex symmetric matrix.

\begin{thm}[Excited state evolution]\label{thm:main}
Let $L_0$ and $L_t = S_t L_0$ be positive Lagrangian subspaces. Consider $Z_t\in\Fn(L_t)$ so that $Z_t = S_tZ_0N_t$ for a Hermitian positive definite matrix $N_t\in\C^{n\times n}$, and denote by $G_t\in\Sp(n,\R)$ the symplectic metric of $L_t$. Define
$$
M_t = \tfrac14 (S_t \bar Z_0)^T G_t (S_t \bar Z_0)
$$
and the polynomials $q_\alpha(x)$, $x\in\C^n$, $\alpha\in \N_0^n$, via the recursion relation 
\begin{equation}\label{eq:rec}
q_0(x)=1 \ ,\qquad q_{\alpha+e_j}(x)=x_j q_{\alpha}(x)-e_j\cdot M_t\nabla q_{\alpha}(x) \ ,\qquad j=1,\ldots,n \ .
\end{equation}
Then, we have for any $\alpha\in \N_0^n$ 
\begin{equation} \label{eq:evolex}
\varphi_{\alpha}(t)= \frac{\ue^{\beta_t}}{\sqrt{\alpha!}} \,q_{\alpha}(N_tA^{\dagger}(Z_t))\varphi_{0}(Z_t) \ ,
\end{equation}
where $\varphi_{\alpha}(0) = \varphi_\alpha(Z_0)$.
\end{thm}
Before entering the proof, we briefly examine the special case of Hermitian time evolution. In this case $S_t$ is real and we can choose $Z_t=S_tZ_0$. 
Then, $N_t = \Id_n$ and $\beta_t=0$, and Proposition~\ref{prop:Z} implies $M_t=\frac14 Z_t^T G_t Z_t = \tfrac{\ui}{4}Z_t^T \Omega Z_t=0$. 
In summary, 
$$
\varphi_{\alpha}(t) = \frac{1}{\sqrt{\alpha!}}A_\alpha^\dagger(Z_t)\varphi_0(Z_t) = \varphi_\alpha(Z_t) \ ,
$$
which is of course also directly implied by $A^\dagger(\bar S_t Z_0) = A^\dagger(Z_t)$, see Lemma~\ref{lem:A_transport} or \cite{Hag80} and \cite{Hag98}. In the more general non-Hermitian case we observe that the time evolution activates lower order states. 
Equation (\ref{eq:evolex}) can be interpreted as an expansion of the propagated state into the basis defined by $Z_t$,
$$
\varphi_{\alpha}(t) = \ue^{\beta_t} \sum_{|k| \leq |\alpha|} a_k(t) \varphi_k(Z_t) \ ,
$$
where the time-dependent coefficients $a_k(t)\in\C$ can be computed in terms of $N_t$ and the polynomial $q_{\alpha}$.
It is worth emphasising the prominent role played by the matrices $N_t$ and $M_t$. All the information about the effects of the non-Hermiticity on the propagation are encoded in those two matrices.

\begin{proof}
We have by Lemma \ref{lem:expand_A}
$$
e_j^T A^{\dagger}(\bar S_t Z_0)=e_j^TC_t^* A^{\dagger}(Z_t) - e_j^TD_t^TA(Z_t) =: \hat u_j - \hat v_j \ .
$$
Then,
$$
A_{\alpha}^{\dagger}(\bar S_t Z_0)=\prod_{j=1}^n (\hat u_j-\hat v_j)^{\alpha_j} \ , 
$$
and in particular 
$A_{\alpha+e_j}^{\dagger}(\bar S_t Z_0)=(\hat u_j-\hat v_j)A_{\alpha}^{\dagger}(\bar S_t Z_0)$.
If we apply $A_{\alpha}^{\dagger}(\bar S_t Z_0)$ to $\varphi_0(Z_t)\in I(L_t)$, then we can use that $\hat v_j\varphi_0(Z_t)=0$. We then commute 
all $\hat v_j$ to the right of the $\hat u_j$ and obtain that 
$$
A_{\alpha}^{\dagger}(\bar S_t Z_0)\varphi_0(Z_t)=q_{\alpha}(\hat u_1, \cdots ,\hat u_n) \varphi_0(Z_t) \ ,
$$ 
where $q_{\alpha}(x)$ is a polynomial in $n$ variables. Our aim is now to derive a recursion relation for~$q_{\alpha}(x)$. Let us define a matrix $M=M_t\in\C^{n\times n}$ by 
$$
M_{ij}:=[\hat v_i,\hat u_j] \ .
$$
Then we have $\hat v_j \hat u_i^k=[\hat v_j,\hat u_i^k]+\hat u_i^k\hat v_j=M_{j,i}k\hat u_i^{k-1} +\hat u_i^k\hat v_j$ and for any polynomial 
$p(\hat u_1, \cdots ,\hat u_n)$
$$
[\hat v_j, p(\hat u_1, \cdots ,\hat u_n)]=(e_j^T M_t\nabla p)(\hat u_1, \cdots ,\hat u_n) \ .
$$
We therefore find
$$
\hat v_j q_{\alpha}(\hat u_1, \cdots ,\hat u_n)=e_j^T M_t\nabla q_{\alpha}(\hat u_1, \cdots ,\hat u_n)+ q_{\alpha}(\hat u_1, \cdots ,\hat u_n) \hat v_j \ .
$$
Using $q_{\alpha+e_j}(\hat u_1, \cdots ,\hat u_n)=(\hat u_j -\hat v_j)q_{\alpha}(\hat u_1, \cdots ,\hat u_n) $ and $\hat v_j\varphi_0(Z_t)=0$ we get 
$$
q_{\alpha+e_j}(\hat u_1, \cdots ,\hat u_n)\varphi_0(Z_t)=\hat u_jq_{\alpha}(\hat u_1, \cdots ,\hat u_n)\varphi_0(Z_t)-e_j^T M_t\nabla q_{\alpha}(\hat u_1, \cdots ,\hat u_n)\varphi_0(Z_t) \ , 
$$
which is the recursion relation \eqref{eq:rec}. It remains to compute $M_t$. We obtain from Lemma \ref{lem:lin_op} that 
$$
[\hat v_i,\hat u_j] = [A(Z_tD_te_i),A^\dagger(Z_tC_te_j)] = \tfrac{\ui}{2} (Z_tD_te_i)^T\Omega(\overline{Z_tC_te_j}) = e_i^T(D_t^T \bar C_t)e_j \ .
$$
Hence, by Lemma~\ref{lem:expand_A} and Proposition~\ref{prop:Z}

\begin{eqnarray*}
M_t &=& D_t^T \bar C_t = \tfrac{1}{4} Z_0^*S_t^T\Omega \bar Z_t Z_t^T \Omega^T S_t \bar Z_0 = 
\tfrac{1}{4} (S_t \bar Z_0)^T (G_t+\ui\Omega)(S_t\bar Z_0)\\
&=& \tfrac{1}{4} (S_t \bar Z_0)^T G_t (S_t \bar Z_0) \ .
\end{eqnarray*}
We observe that $M_t=M_t^T$ and notice that $(\hat u_1, \cdots,  \hat u_n)^T= C_t^* A^{\dagger}(Z_t)$. Moreover, 
$$
C_t = \tfrac\ui2 N_t^* Z_0^* (S_t^*\Omega^T \bar S_t) Z_0 = \frac\ui2 N_t Z_0^*\Omega^T Z_0 = N_t \ ,
$$
since $S_t$ is symplectic and $Z_0$ normalised. 
\end{proof}

Applying a linear combination of powers of $A^\dagger(Z_t)$ to the normalised Gaussian $\varphi_0(Z_t;x)$ produces multivariate polynomials in $x$ that can be described by a recursion relation of the type encountered above. 

\begin{cor}[Polynomial prefactor]\label{cor:pol}
Let $L_0$ and $L_t := S_t L_0$ be positive Lagrangian subspaces. Let $Z_t\in\Fn(L_t)$ so that $Z_t= S_tZ_0N_t$ for a Hermitian positive definite matrix $N_t\in\C^{n\times n}$ and set
$$
Z_t = \begin{pmatrix}P_t\\ Q_t\end{pmatrix} \ .
$$ 
Denote by $G_t\in\Sp(n,\R)$ the symplectic metric of $L_t$. Define 
$$
M_t = \tfrac14 (S_t \bar Z_0)^T G_t (S_t \bar Z_0)\quad\text{and}\quad \widetilde M_t = M_t + N_t Q_t^{-1} \overline Q_t \overline N_t \ .
$$
We then have for any $\alpha\in \N_0^n$ 
\begin{equation}
\label{eq:identity}
\varphi_{\alpha}(t;x) =  \frac{\ue^{\beta_t}}{\sqrt{\alpha!}}\,p_{\alpha}\left(\sqrt{\tfrac{2}{\veps}}\,N_t Q_t^{-1} x\right) \, \varphi_0(Z_t;x)
\end{equation}
where the polynomials $p_\alpha(x)$, $x\in\C^n$, satisfy the recursion relation 
$$
p_0(x)=1 \ ,\qquad p_{\alpha+e_j}(x)=x_j p_{\alpha}(x)-e_j\cdot \widetilde M_t\nabla p_{\alpha}(x) \ ,\qquad j=1,\ldots,n \ .
$$
\end{cor}
\begin{proof}
We first compute  
\begin{eqnarray*}
N_t A^\dagger(Z_t)\varphi_0(Z_t;x) &=&\tfrac{\ui}{\sqrt{2\veps}} N_t\left(P_t^* Q_t -Q_t^* P_t\right) Q_t^{-1}x  \varphi_0(Z_t;x) \\
&=& \sqrt{\tfrac{2}{\veps}}N_t Q_t^{-1}x\varphi_0(Z_t;x)\\
&=& y_t \varphi_0(Z_t;x) \ ,
\end{eqnarray*}
with $y_t = \sqrt{\tfrac{2}{\veps}}\,N_t Q_t^{-1}x$ and where we have used the normalisation $Z_t^*\Omega Z_t = Q_t^*P_t - P_t^*Q_t = 2 \ui \Id_n$. 
This motivates the ansatz
$$
q_\alpha(N_tA^\dagger(Z_t)) \varphi_0(Z_t;x) =: p_\alpha(y_t)\varphi_0(Z_t;x) \ .
$$
The gradient formula of Lemma~\ref{lem:gradient} implies
$$
p_{\alpha+e_j}(y_t)\varphi_0(Z_t;x) = e_j\cdot N_t A^\dagger(Z_t) \left(p_\alpha(y_t)\varphi_0(Z_t;x)\right) - 
e_j\cdot M_t(\alpha_j p_{\alpha-e_j}(y_t))_{j=1}^n \varphi_0(Z_t;x) \ .
$$
We compute
\begin{eqnarray*}
N_tA^\dagger(Z_t)(p_\alpha(y_t)\varphi_0(Z_t;x)) &=& 
p_\alpha(y_t) N_t A^\dagger(Z_t)\varphi_0(Z_t;x) - \tfrac{\ui}{\sqrt{2\veps}}\varphi_0(Z_t;x) N_t Q_t^* \hat p \,p_\alpha(y_t)\\
&=& y_t p_\alpha(y_t)\varphi_0(Z_t;x) -N_tQ_t^* Q_t^{-T}N_t^T (\nabla p_\alpha)(y_t)\varphi_0(Z_t;x)
\end{eqnarray*}
so that
$$
p_{\alpha+e_j}(y_t) = {y_t}_j p_\alpha(y_t) - e_j\cdot (M_t + N_tQ_t^* Q_t^{-T}N_t^T)(\alpha_j p_{\alpha-e_j}(y_t))_{j=1}^n \ .
$$
Since $Q_tQ_t^*$ is real symmetric, we have $Q_t^* Q_t^{-T}= Q_t^{-1} \overline Q_t$ and $\widetilde M_t$ is symmetric.  
\end{proof}

\subsection{Dynamics of the center} 

Repeating the calculations of Lemma~\ref{lem:A_transport}, the time-evolution of the centered ladder operators reads
\begin{equation}\label{eq:low_evo}
t\mapsto  A(S_t l,S_t z) \quad\text{and} \quad t\mapsto A^\dagger(\bar S_t l,\bar S_t z)
\end{equation}
for all $l,z\in\C^n\oplus\C^n$. We assume that $L_0$ and $L_t = S_t L_0$ are positive Lagrangian subspaces and consider the complex structure $J_t\in\Sp(n,\R)$ of the Lagrangian $L_t$. We then know by Theorem~\ref{thm:real} that a real projection of the center $S_tz$ does not change the ladder operator if we parametrise by the Lagrangian $L_t$, that is,
$$
A(S_t l,S_t z) = A(S_t l,P_{J_t}(S_tz))
$$
for all $l\in L_0$ and $z\in\C^n\oplus\C^n$. The dynamics of the projected center are easily inferred from the Riccati equations for the complex structure $J_t$. They reflect the metriplectic structure of equation~(\ref{eq:metri_schro}) on the finite dimensional level.  

\begin{cor}[Projected dynamics]\label{cor:metriplectic}
Let $L_0$ and $L_t=S_tL_0$ be positive Lagrangian subspaces. Denote by $G_t,J_t\in\Sp(n,\R)$ the symplectic metric and the complex structure of $L_t$, respectively. Let $z_0\in\R^n\oplus\R^n$. Then, $z_t := P_{J_t}(S_t z_0)\in\R^n\oplus\R^n$ satisfies
\begin{equation}\label{eq:real_dyn}
\dot z_t = \Omega\Re H_t z_t +G_t^{-1} \Im H_t z_t \ .
\end{equation}
\end{cor}

\begin{proof}
We differentiate $z_t = \Re(S_t z) + J_t\Im(S_t z)$ so that Theorem~\ref{thm:ric} implies
\begin{eqnarray*}
\dot z_t &=& \Re(\Omega H_t S_t z) + \dot J_t \Im(S_t z) + J_t \Im(\Omega H_tS_t z)\\
&=& \Omega\Re H_t\Re(S_t z) - \Omega\Im H_t\Im(S_t  z)\\
&& + \left(\Omega\Re H_t J_t - J_t\Omega \Re H_t + \Omega \Im H_t + J_t\Omega \Im H_t J_t\right)\Im(S_t z)\\
&& + \,J_t\Omega \Im H_t\Re(S_t z) + J_t\Omega\Re(H_t)\Im(S_t z)\\
&=& \Omega\Re H_t z_t + J_t\Omega \Im H_t z_t \ .
\end{eqnarray*}
Moreover, $\Omega J_t=G_t$ gives $J_t\Omega = \Omega^T G_t\Omega = G_t^{-1}$.
\end{proof}

The time evolution of coherent states with real projected center resembles the one of Hermitian dynamics, however, with a phase factor determined by the action integral of the Hamiltonian $\ham_t$ along the real projected trajectory.

\begin{prop}[Coherent state evolution]\label{prop:coh_z}
Let $L_0\subset\C^n\oplus\C^n$ be a positive Lagrangian subspace, $Z_0\in\Fn(L_0)$ and $z_0\in\R^n\oplus\R^n$. 
Let the coherent state $\varphi_0(Z_0,z_0)=:\varphi_0(0)$ be given by \eqref{eq:coh_cent}. If the Lagrangian $L_t = S_t L_0$ is positive for $t\in [0,T[$, then 
$$
\varphi_0(t) = \ue^{\frac{\ui}{\veps}\alpha_t(z_0)}\ \varphi_0(S_tZ_0,z_t)
= \ue^{\frac{\ui}{\veps}\alpha_t(z_0) + \beta_t}\varphi_0(Z_t,z_t)
$$
for all $t\in[0,T[$, where $z_t=:(p_t,q_t)\in\R^d\oplus\R^d$ is defined by \eqref{eq:real_dyn}, $\beta_t$ is the factor derived in Proposition \ref{prop:coh} and 
\begin{equation} \label{eq:alpha}
\alpha_t(z_0) := \int_0^t \left(\dot q_\tau\cdot p_\tau - \ham_\tau(z_\tau)\right) d\tau
\end{equation}
denotes the associated action integral of the Hamiltonian $\ham_t$ along $z_t$.
\end{prop}

\begin{proof}
Starting from the initial value $0=A(Z_0,z_0) \varphi_0(Z_0,z_0)$, we find for the time evolution using the propagated lowering operator \eqref{eq:low_evo},
\begin{eqnarray*}
0  = A(S_t Z_0, z_t)  \varphi_0(t) \ .
\end{eqnarray*}
Therefore, $\varphi_0(Z_0,z_0)\in I(L_0,z_0)$ implies $\varphi_0(t)\in I(L_t,z_t)$, and hence there exists $c_t\in\C$ with 
$\varphi_0(t)=c_t \cdot \varphi_0(S_tZ_0,z_t)$. It remains to determine $c_t$. We denote
$$
S_t Z_0 = \begin{pmatrix}P_t\\ Q_t\end{pmatrix} \ ,\qquad H_t = \begin{pmatrix}H_{pp}&H_{pq}\\ H_{qp} & H_{qq}\end{pmatrix} \ .
$$
Computing $\ui\veps\partial_t \varphi_0(t)$ we obtain
$$
\ui\veps\dot c_t/c_t + \ui\veps\left(\pa_t\det(Q_t)^{-1/2}\right)\det(Q_t)^{1/2}  + \ui\veps\pa_t\left( 
\tfrac{\ui}{2\veps}(x-q_t)\cdot B_t(x-q_t) + \tfrac{\ui}{\veps}p_t\cdot(x-q_t)\right)
$$
times $\varphi_0(t)$. We sort this second order polynomial in powers of $(x-q_t)$ and keep the constant terms, that is,
\begin{equation}\label{eq:time}
\ui\veps\dot c_t/c_t -\tfrac{\ui\veps}{2} \tr(\pa_t Q_t Q_t^{-1})  + p_t \cdot \dot q_t \ ,
\end{equation}
where we have used Jacobi's determinant formula $(\pa_t \det Q_t)/\det Q_t=\tr(\pa_t Q_t Q_t^{-1})$. Next we compute 
\begin{eqnarray*}
\hat \ham_t \varphi_0(t) &=& \tfrac12 (\hat z\cdot H_t\hat z)\varphi_0(t)
\;=\;\tfrac12 \hat p\cdot\left( \left(H_{pp}B_t(x-q_t)+ H_{pp}p_t + H_{pq}x \right)\varphi_0(t)\right)\\
 &&+ 
\tfrac12 x\cdot\left(H_{qp}B_t(x-q_t)+H_{qp}p_t + H_{qq}x)\right)\varphi_0(t)
\end{eqnarray*}
Therefore $\hat \ham_t \varphi_0(t)$ is a second order polynomial in $(x-q_t)$ times $\varphi_0(t)$, and the constant terms amount to
\begin{equation}\label{eq:space}
\tfrac{\veps}{2\ui} \tr(H_{pp}P_tQ_t^{-1} + H_{pq}) + \ham_t(z_t) \ .
\end{equation}
Since $\partial_t Q_t = H_{pq}Q_t + H_{pp}P_t$ and $\partial_t Q_t Q_t^{-1} = H_{pq} + H_{pp} P_tQ_t^{-1}$, the matching of the terms in \eqref{eq:time} and \eqref{eq:space} gives
$$
\ui\veps\dot c_t/c_t + p_t\cdot \dot q_t = \ham_t(z_t) \ ,
$$
which is solved by the exponential of the action integral $c_t = \ue^{\frac{\ui}{\veps}\alpha_t(z_0)}$.
\end{proof}

Our previous results on excited state propagation, that is, Theorem~\ref{thm:main} and Corollary~\ref{cor:pol}, describe the time evolution of $\varphi_\alpha(t)$ with
$$
\varphi_\alpha(0) = \varphi_\alpha(Z_0,z_0) \ ,\qquad\alpha\in\N_0^{n} \ ,
$$
for the case $z_0=0$ in terms of multivariate polynomials. Essentially, these results stay the same when considering nonzero $z_0\in\R^n\oplus\R^n$. We only have to record the evolution of the center and add the corresponding action integral. 

\begin{thm}[Excited state evolution]\label{thm:main_center}
Let $L_0$ and $L_t = S_t L_0$ be positive Lagrangian subspaces. Let $z_0\in\R^n\oplus\R^n$ and $Z_t\in\Fn(L_t)$ so that $Z_t= S_tZ_0N_t$ for a Hermitian positive definite matrix $N_t\in\C^{n\times n}$. Set
$$
Z_t = \begin{pmatrix}P_t\\ Q_t\end{pmatrix}
$$ 
and denote by $G_t\in\Sp(n,\R)$ the symplectic metric of $L_t$. Define 
$$
M_t = \tfrac14 (S_t \bar Z_0)^T G_t (S_t \bar Z_0)\quad\text{and}\quad \widetilde M_t = M_t + N_t Q_t^{-1} \overline Q_t \overline N_t \ .
$$
Then, we have for any $\alpha\in \N_0^n$ 
\begin{eqnarray*}
\varphi_{\alpha}(t;x) &=&  
\frac{\ue^{\frac{\ui}{\veps}\alpha_t(z_0)+\beta_t}}{\sqrt{\alpha !}}\,q_{\alpha}(N_t A^\dagger(Z_t,z_t)) \, \varphi_0(Z_t,z_t;x)\\
&=&
\frac{\ue^{\frac{\ui}{\veps}\alpha_t(z_0)+\beta_t}}{\sqrt{\alpha !}}\,p_{\alpha}\left(\sqrt{\tfrac{2}{\veps}}\,N_t Q_t^{-1} (x-q_t)\right) \, \varphi_0(Z_t,z_t;x)
\end{eqnarray*}
where $z_t=(p_t,q_t)\in\R^n\oplus\R^n$ is defined by \eqref{eq:real_dyn} and $\alpha_t(z_0)$ is the action integral (\ref{eq:alpha}) of $\ham_t$ along the trajectory $z_t$. The polynomials $q_\alpha(x) = r_\alpha(x;M_t)$ and $p_\alpha(x) = r_\alpha(x;\widetilde M_t)$ satisfy the recursion relations 
$$
r_0(x;M)=1 \ ,\qquad r_{\alpha+\ue_j}(x;M)=x_j r_{\alpha}(x;M)-e_j\cdot M\nabla r_{\alpha}(x;M) \ ,\qquad j=1,\ldots,n \ ,
$$
with $M = M_t$ and $M=\widetilde M_t$, respectively.
 \end{thm}

The time evolution of almost all the constitutive elements of Theorem~\ref{thm:main_center} can be described by ordinary differential equations: First, there is
the Riccati equation of Theorem~\ref{thm:ric} for the symplectic metric $G_t$, that can be solved together with the equation for the loss or gain parameter $\beta_t$,
$$
\partial_t\beta_t = \tfrac14\tr(G_t^{-1}\Im H_t) \ ,\qquad \beta_0 = 0 \ .
$$
Second, there is the metricplectic equation of Corollary~\ref{cor:metriplectic} for the real center $z_t$, together with the corresponding action integral $\alpha_t(z_0)$. Finally, for the normalised Lagrangian frame $Z_t = S_tZ_0N_t$, we find the equation
$$
\partial_t Z_t = \Omega H_t Z_t + Z_t N^{-1}_t \partial_t N_t \ .
$$
which contains the time derivative of the normalising matrix $N_t$. We will illustrate in the following section how one can determine $N_t$ for explicit one-dimensional examples.

\section{Examples}\label{sec:swanson}
As examples we investigate the dynamics of the following model systems: the one-dimensional Davies--Swanson oscillator
$$
\hat \ham_S = \frac{\omega_0}{2} (\hat{p}^2+\hat{q}^2) - \frac{\ui\delta}{2}(\hat{p}\hat{q} + \hat{q}\hat{p}) = \tfrac12\Op[z\cdot H_S z]
$$
defined by the complex symmetric matrix 
$$
H_S = \begin{pmatrix} \omega_0 & -\ui \delta \\ -\ui \delta & \omega_0 \end{pmatrix},\qquad \omega_0,\delta>0 \ ,
$$
whose imaginary part is a real symmetric matrix with eigenvalues $\pm\delta$ and a diffusion equation of the form
$$
\partial_t \rho = \alpha \Delta \rho \ , \qquad H_D =-2\ui \alpha \begin{pmatrix} \Id_n & 0 \\ 0 & 0 \end{pmatrix},\qquad \alpha \in \C \ 
$$
in dimension $n=1$ and $n=2$. For the Davies--Swanson oscillator the spectrum and transition elements have been computed  \cite{Dav99b, Sw04} as well as the dynamics of coherent states \cite{GrKoRuSch15}. It is our aim here to complement the picture by propagating excited wavepackets. Our general approach for the diffusion equation, i.e. taking complex $\alpha$ into account, allows us to compare in particular the dynamics of the free Schr\"odinger equation ($\alpha=-\ui$) and the heat equation ($\alpha=1$).

\subsection{One-dimensional systems}
For one-dimensional systems, the results of Theorem~\ref{thm:main} simplify, since the normalisation of Lagrangian frames just involves the inversion of a positive real number. Starting with a positive Lagrangian subspace $L_0 = \spann\{l_0\}$ 
spanned by a normalised vector $l_0\in\C\oplus\C$, we set
$$
l_t := S_t l_0 n_t = \begin{pmatrix}p_t \\ q_t\end{pmatrix}\quad\text{with}\quad n_t^{-2} = h(S_t l_0,S_t l_0)>0
$$
to obtain a normalised Lagrangian frame $l_t\in\C\oplus\C$ of the time evolved subspace $L_t=\bar S_tL_0$. The gain or loss parameter $\beta_t$ is 
then simply given by
\[
\ue^{\beta_t} = \sqrt{n_t}\ .
\]
In order to describe the propagation of excited wavepackets we use
\begin{eqnarray*}
m_t =  d_t \overline c_t = n_t^2\, h(S_t l_0,S_t \bar l_0) \ .
\end{eqnarray*}
For notational convenience, we restrict ourselves to the case $z_0=0$. For non-vanishing centers $z_0\in\R\oplus\R$, there is an additional multiplicative factor due to the complex-valued action integral $\alpha_t(z_0)$. According to Proposition~\ref{prop:coh} and equation \eqref{eq:first}, we obtain
\begin{eqnarray*}
\varphi_0(t) &=& \ue^{\beta_t} \varphi_0(l_t) \ ,\\
\varphi_1(t) &=& \ue^{\beta_t} n_t \varphi_1(l_t) \ ,\\
\varphi_2(t) &=& \ue^{\beta_t}\left(n_t^2 \varphi_2(l_t) -  \tfrac{1}{\sqrt{2}}m_t \varphi_0(l_t) \right) \ ,\\
\varphi_3(t) &=& \ue^{\beta_t}\left(n_t^3 \varphi_3(l_t) -  \tfrac{3}{\sqrt{6}} m_t n_t \varphi_1(l_t) \right)
\end{eqnarray*}
for the coherent and the first three excited state, respectively. Their norms evolve according to
\begin{equation} \label{eq:Norms}
\begin{aligned}
\|\varphi_0(t)\| &= \ue^{\beta_t}\ , \\
\|\varphi_1(t)\| &= \ue^{\beta_t}n_t \ , \\
\|\varphi_2(t)\| &= \ue^{\beta_t} \sqrt{n_t^4 + \tfrac12 |m_t|^2} \ , \\  
\|\varphi_3(t)\| &= \ue^{\beta_t}n_t\sqrt{n_t^4 + \tfrac32 |m_t|^2}\ .
\end{aligned}
\end{equation}
For the whole orthonormal basis Theorem~\ref{thm:main} provides 
\begin{equation}\label{eq:expansion}
\varphi_k(t) = \frac{\ue^{\beta_t}}{\sqrt{k!}} q_k(n_t A^\dagger(l_t)) \varphi_0(l_t) \ ,\qquad k\in\N_0 \ ,
\end{equation}
where the univariate polynomials $q_k$ satisfy the recursion relation
\begin{equation}\label{eq:hermite}
q_0(x) = 1 \ ,\qquad q_{k+1}(x) = xq_k(x) - m_t\, q_k'(x) \ ,\quad k\in\N_0 \ .
\end{equation}
These polynomials are Hermite polynomials with time dependent scaling according to the complex number $m_t$. 
Using the monomial expansion of these polynomials, we can rewrite the expansion in \eqref{eq:expansion} explicitly in terms of the propagated basis functions $\varphi_k(l_t)$.

\begin{cor}[Explicit expansion]\label{cor:norm} Let $l_0\in\C\oplus\C$ so that $L_0=\spann\{l_0\}$ and $L_t = \bar S_t L_0$ are positive Lagrangian subspaces. Let $n_t>0$ so that $l_t=S_t l_0 n_t$ is normalised according to $h(l_t,l_t)=1$. Then, for all $k\in\N_0$, 
$$
\varphi_k(t) = \ue^{\beta_t}\sum_{j=0}^k \sqrt{\frac{j!}{k!}}a_{kj} n_t^j \varphi_j(l_t) \ ,
$$
where $a_{kj}\in\C$ are the coefficients of the monomial expansion $q_k(x) = \sum_{j=0}^k a_{kj} x^j$ of the polynomials $q_k(x)$ defined by the recursion relation  \eqref{eq:hermite}.
\end{cor}

\begin{proof}
Since $\frac{1}{\sqrt{j!}}A^\dagger_j(l_t)\varphi_0(l_t) = \varphi_j(l_t)$ for all $j$, we have
$$
\varphi_k(t)= \frac{\ue^{\beta_t}}{\sqrt{k!}}\sum_{j=0}^k a_{kj} n_t^j A_j^\dagger(l_t)\varphi_0(l_t) = \ue^{\beta_t}\sum_{j=0}^k \sqrt{\frac{j!}{k!}}a_{kj} n_t^j \varphi_j(l_t) \ .
$$
\end{proof}
Since all Hermite functions $\varphi_j(l_t)$, $j=1, \ldots, k$, can be written as a polynomial of degree~$j$ times the coherent state $\varphi_0(l_t)$,  the above expansion exhibits the same structure. Due to Corollary \ref{cor:pol} the polynomial part of $\varphi_k(t)$ is again a Hermite polynomial scaled by the factor 
\[
\widetilde m_t = m_t +  n_t^2\, q^{-1}_t \bar q_t\ .
\]  
An explicit calculation shows that
\begin{eqnarray*}
\varphi_1(t,x) &=& y_t \, \varphi_0(t, x) \ ,\\
\varphi_2(t,x) &=& \tfrac{1}{\sqrt{2}} \left(y^2_t - \widetilde m_t\right) \varphi_0(t,x)  \ ,\\
\varphi_3(t,x) &=& \tfrac{1}{\sqrt{6}} \left(y^3_t - 3 \widetilde m_t y_t \right) \varphi_0(t,x)
\end{eqnarray*}
with the scaled variable 
\[
y_t = \sqrt{\tfrac{2}{\veps}} n_t q^{-1}_t x\ .
\] 
In particular, we find that the roots of $\varphi_2$ and $\varphi_3$, except for the origin, depend on $n^{-1}_t q_t \widetilde m_t^{1/2}$. In Appendix \ref{ssec: 1Droots} we present a similar study of the roots of one-dimensional wavepackets in the stationary case.
\subsection{Norm evolution for the Davies--Swanson oscillator}

Applying the one-dimensional formulas to our first example, the Davies--Swanson oscillator, we start by examining the classical Hamiltonian system 
$$
\dot S_t = \Omega H_S S_t \ ,\qquad S_0 = \Id_2 \ .
$$ 
Its solution $S_t = \exp(t\Omega H_S)$ exists for all times $t\in\R$. Setting $\omega^2 :=\omega^2_0 + \delta^2$, we observe $(\Omega H_S)^2 = -\omega^2\Id_2$ and consequently
$$
(\Omega H_S)^{2k} = (-1)^k \omega^{2k} \Id_2 \ ,\qquad (\Omega H_S)^{2k+1} = (-1)^k \omega^{2k} \Omega H_S \ ,\qquad k\ge0 \ .
$$ 
Therefore,  
\begin{eqnarray*}
S_t &=& \sum^{\infty}_{k=0} (-1)^k \frac{t^{2k}}{(2k)!} \omega^{2k}\Id_2 + \sum^{\infty}_{k=0} (-1)^k \frac{t^{2k+1}}{(2k+1)!} \omega^{2k}  \Omega H_S\\ 
&=& \cos(t \omega) \Id_2 + \tfrac{1}{\omega} \sin(t \omega) \Omega H_S \ .
\end{eqnarray*}
This formula for $S_t$ allows to explicity compute the time-intervals for which a particular initial Lagrangian subspace stays positive. 
 
\begin{lem}[Positive Lagrangian subspace]\label{lem:swanson} Let $L_0 = \spann\{l_0\}$ with $l_0=(1,-\ui)\in\C\oplus\C$ and consider $L_t = S_t L_0$. If $\omega_0>\delta$, then $L_t$ is a positive Lagrangian subspace for all $t\in\R$. Otherwise, $L_t$ is positive for $t\in[0,T[$ with
$$
T := \frac{1}{2\omega}\arccos\!\left(-\frac{\omega_0^2}{\delta^2}\right) \ .
$$ 
\end{lem}

\begin{proof}
We first compute 
\begin{eqnarray*}
S_t^*\Omega S_t &=& \left(\cos(t\omega)\Id_2 + \tfrac1\omega\sin(t\omega)\Omega H_S\right)^*\Omega\left(\cos(t\omega)\Id_2 + \tfrac1\omega\sin(t\omega)\Omega\bar H_S\right)\\
&=& 
\cos^2(t\omega)\Omega + \tfrac{1}{\omega^2}\sin^2(t\omega) \bar H_S\Omega H_S - \tfrac{2\ui}{\omega} \cos(t\omega)\sin(t\omega)\Im H_S \ .
\end{eqnarray*}
and
$$
\bar H_S\Omega H_S = \begin{pmatrix}2\ui\delta\omega_0 & \delta^2-\omega_0^2\\ \omega_0^2-\delta^2 & -2\ui\delta\omega_0\end{pmatrix},\qquad
\Im H_S = \begin{pmatrix}0 & -\delta\\ -\delta & 0\end{pmatrix} \ .
$$
This implies for all normalised vectors $l\in\C\oplus\C$ with $h(l,l) = \frac{\ui}{2}l^*\Omega^T l = 1$ that  

\begin{eqnarray*}
h(S_t l,S_t l)
&=&
\cos^2(t\omega) - \tfrac{\ui}{2\omega^2}\sin^2(t\omega) l^* \bar H_S\Omega H_Sl -\tfrac1\omega \cos(t\omega)\sin(t\omega)l^*\Im H_S l \ .
\end{eqnarray*}
With $l=(l_1,l_2)$
\begin{eqnarray*}
l^* \bar H_S\Omega H_S l &=& 2\ui\delta\omega_0(|l_1|^2-|l_2|^2) -2\ui(\omega_0^2-\delta^2)\Im(\overline l_1 l_2) \ ,\\
l^*\Im H_S l &=& -2\delta \Re(\overline l_1 l_2) \ ,
\end{eqnarray*}
In one dimensional systems the normalisation of $l$ is equivalent to $\Im(l_1 \bar l_2) = 1$, so we can replace imaginary part in the equation above. However, there is no relation between $l_1$ and $l_2$ in general, so we cannot simplify this further. 

For the particular vector $l_0=(1,-\ui)$ we obtain
$$
h(S_t l_0,S_t l_0) = \cos^2(t\omega) +\frac{\omega_0^2-\delta^2}{\omega_0^2+\delta^2}\sin^2(t\omega) = 
1 - \frac{\delta^2}{\omega^2}\left(1-\cos(2t\omega)\right) \ .
$$
This function is positive for all $t\in\R$, if $\omega_0\ge\delta$. Otherwise positivity holds on $[0,T[$. 
\end{proof}

We consider $l_0 = (1,-\ui)$ and work for times $t$ so that the Lagrangian subspace $L_t = S_t L_0$ is positive. We obtain the normalisation factor $n_t>0$ with
\begin{align*}
n_t^{-2} 
&= h(S_t l_0,S_t l_0) = 1 - \frac{\delta^2}{\omega^2}\left(1-\cos(2t\omega)\right)\\
&= \omega^{-2} \left(\omega_0^2 + \delta^2 \cos(2t\omega)\right)
\end{align*}
and the real-valued gain or loss factor $\beta_t$ according to
\[
\ue^{\beta_t} = n_t^{1/2} = \omega^{1/2}\left(\omega_0^2 + \delta^2 \cos(2t\omega)\right)^{-1/4}\ .
\]
For the polynomial recursion \eqref{eq:hermite} we also have to compute
$$
m_t = n_t^2\, h(S_t l_0,S_t \bar l_0) \ .
$$
Repeating a part of the calculations of the proof of Lemma~\ref{lem:swanson}, we obtain for all $l\in\C\oplus\C$
$$
l^* \bar H_S\Omega H_S \bar l = 2\ui\delta\omega_0\left(\overline l_1^2- \overline l_2^2\right) \quad\text{and}\quad
l^*\Im H_S \bar l = -2\delta\, \overline l_1 \overline l_2
$$
so that
$$
m_t = \frac{2\delta}{\omega} n_t^2\sin(t\omega)\left(\frac{\omega_0}{\omega}\sin(t\omega)+\ui\cos(t\omega)\right) \ .
$$
Having derived explicit formulas for the time evolution of the parameters,  
we now use the formulas of \eqref{eq:Norms} for the norm evolution for the coherent state and the first three excited states.
As expected, all four norms considerably depart from unity, the more highly excited the state, the stronger the deviation, see Figure~\ref{fig:normDavies}.

\begin{figure}[h]
\includegraphics[width=\textwidth]{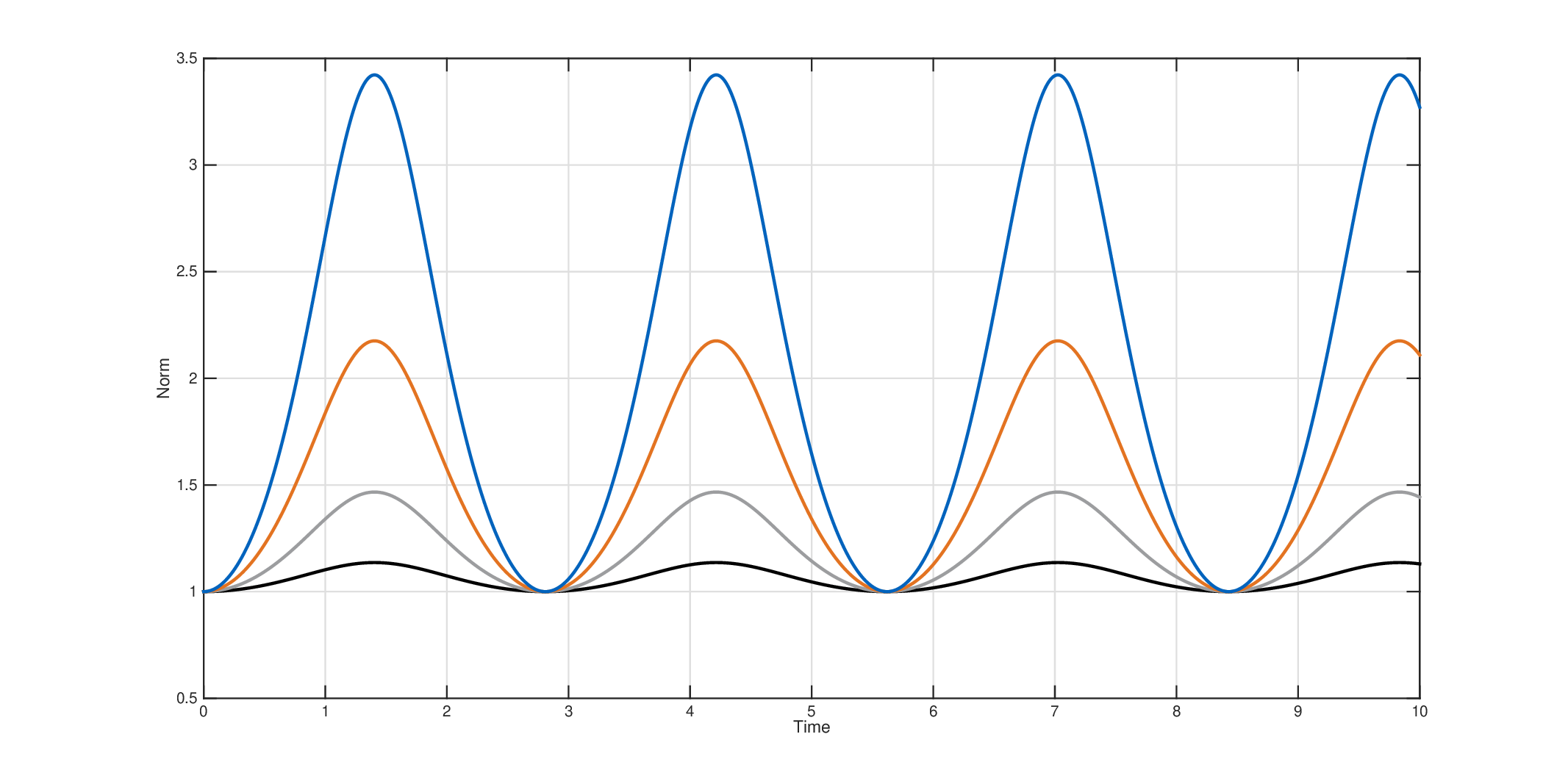}
\caption{Time evolution of $\|\varphi_k(t)\|$ governed by the Davies--Swanson oscillator with initial data defined by $\l_0=(1,-\ui)$. The color coding is black for $k=0$, grey for $k=1$, orange for $k=2$ and blue for $k=3$. The parameters of the oscillator are chosen as $\omega_0=1$ and $\delta=0.5$ such that the solution $\varphi_k(t)$ exists for all $t\in\R$. In all four cases the norms considerably depart from unity, the larger $k$ is, the stronger they deviate.}\label{fig:normDavies}
\end{figure}
\subsection{Evolution of the roots for the diffusion equation} \label{sec:rootsDiff}
Our proceeding for the one-dimensional diffusion equation is analogue to the Davies--Swanson oscillator,
$$
H_D = -2\ui \alpha \begin{pmatrix} 1 & 0 \\ 0 & 0 \end{pmatrix}, \qquad \alpha \in \C \, .
$$
Since $(\Omega H_D)^2 = 0$ we find for the Hamiltonian system $\dot S_t = \Omega H_D S_t$ with $S_0 = \Id_2$ that
$$
S_t = \exp(t\Omega H_D) = \Id_2 + t \, \Omega H_D = \begin{pmatrix} 1 & 0 \\ -2\ui \alpha t & 1\end{pmatrix} \ .
$$ 
for all times $t\in\R$. Since the spectrum of $\Im(H_D)$ is $\sigma= \{-2\Re(\alpha),0\}$ an initial positive Lagrangian subspace stays positive for all $t \geq 0$ if $\Re(\alpha) \geq 0$. In practice $\Re(\alpha)$ measures the directed transfer rate of the medium, i.e. the larger $\Re(\alpha)$ the more transmissible our system is. If $\Re(\alpha) \geq 0$ we study the standard setting where particles are transferred from regions with higher concentration to regions with lower concentration. However, to provide a full theoretical description we investigate also the case $\Re(\alpha)<0$ in the following result. 
\begin{lem}[Positive Lagrangian subspace]\label{lem:diff1}
Let $L_0 = \spann\{l_0\}$ be a positive Lagrangian subspace spanned by a normalised Lagrangian frame $l_0=(p_0, q_0)\in\C\oplus\C$. Then, $L_t =S_t L_0$ is positive for all $t\in\R$ if $\Re(\alpha) \geq 0$ and for $t\in[0,T[$ with
$$
T := \left(-2 \Re(\alpha) |p_0|^2\right)^{-1}
$$ 
if $\Re(\alpha)<0$.
\end{lem}
\begin{proof}
A direct calculation yields
$$
S^*_t \Omega S_t =  \begin{pmatrix} 1 & 2\ui \bar \alpha t \\ 0 & 1\end{pmatrix} \begin{pmatrix} 2\ui \alpha t & -1 \\ 1 & 0\end{pmatrix} =  \begin{pmatrix} 4\ui \Re(\alpha) t & -1 \\ 1 & 0\end{pmatrix}
$$
and thus
$$
n^{-2}_t = h(S_t l_0,S_t l_0) = \frac{1}{2i} \begin{pmatrix} \bar p_0 & \bar q_0\end{pmatrix} \begin{pmatrix} 4\ui \Re(\alpha) t & -1 \\ 1 & 0\end{pmatrix} \begin{pmatrix} p_0 \\  q_0\end{pmatrix} = 1+2 \Re(\alpha)|p_0|^2 t
$$
where we used that the normalisation of $l_0$ implies $\Im(p_0 \bar q_0)=1$.
\end{proof}
In the following we only consider times $t$ such that $L_t$ is a positive Lagrangian subspace. 
The calculation of the normalisation in the previous proof gives
$$
e^{\beta_t} = n^{1/2}_t = \left(1+2 \Re(\alpha)|p_0|^2 t \right)^{-1/4}.
$$
It remains to determine the factors for the polynomial recursion, 
$$
m_t = n_t^2\, h(S_t l_0,S_t \bar l_0) = 2 n_t^2 \Re(\alpha) \bar p_0^2 t, \qquad \widetilde m_t =m_t + n^2_t q^{-1}_t \bar q_t = \frac{\bar q_0 - 2 \ui \alpha t \bar p_0}{q_0-2\ui \alpha t p_0} \ .
$$
As direct example we investigate the heat equation, $\alpha = 1$, again with the special initial value $l_0=(1,-\ui)$. One can easily derive the explicit formulas
$$
n^{-2}_t = 1+2t, \qquad m_t = \frac{2t}{1+2t} \ .
$$
Moreover, $\tilde m_t = \frac{2t-1}{2t+1}$ and the evolution of the roots of the second excited state $\varphi_2(t)$ is given by
\begin{align*}
x_{1/2}(t) &=\pm \, n^{-1}_t q_t \sqrt{\tfrac{1}{2} \widetilde m_t} \\
&= \pm\sqrt{\tfrac{1}{2}(1+2t)(1-2t)} \, .
\end{align*}
Hence, the roots are propagated towards the origin and vanish at $t=0.5$. The behaviour for the roots of the third excited states follows similarly.  
Figures \ref{fig:normHeat} below displays the evolution of the absolute value of the first three excited states $|\varphi_k(t)|$, $k=0,1,2,3$, in the described setting. 

\begin{figure}
\includegraphics[width=0.5\textwidth]{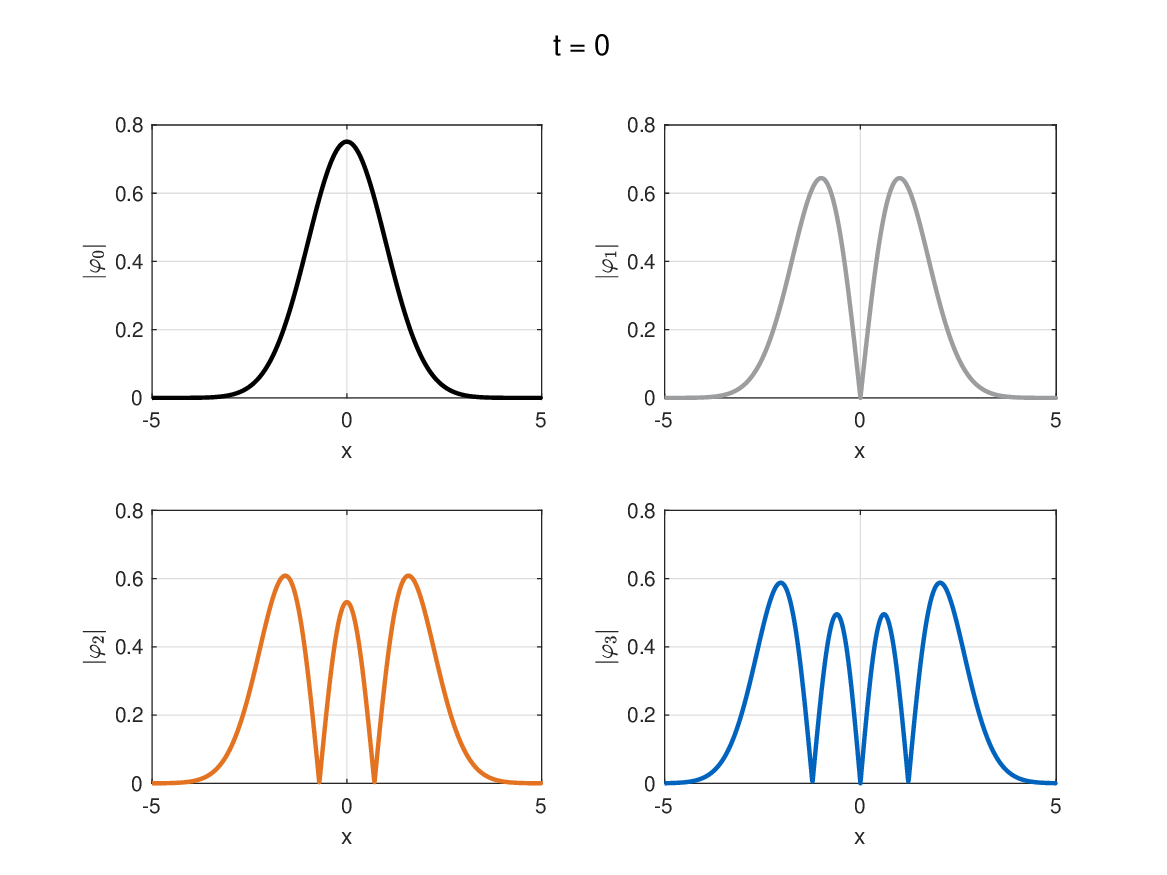}
\includegraphics[width=0.5\textwidth]{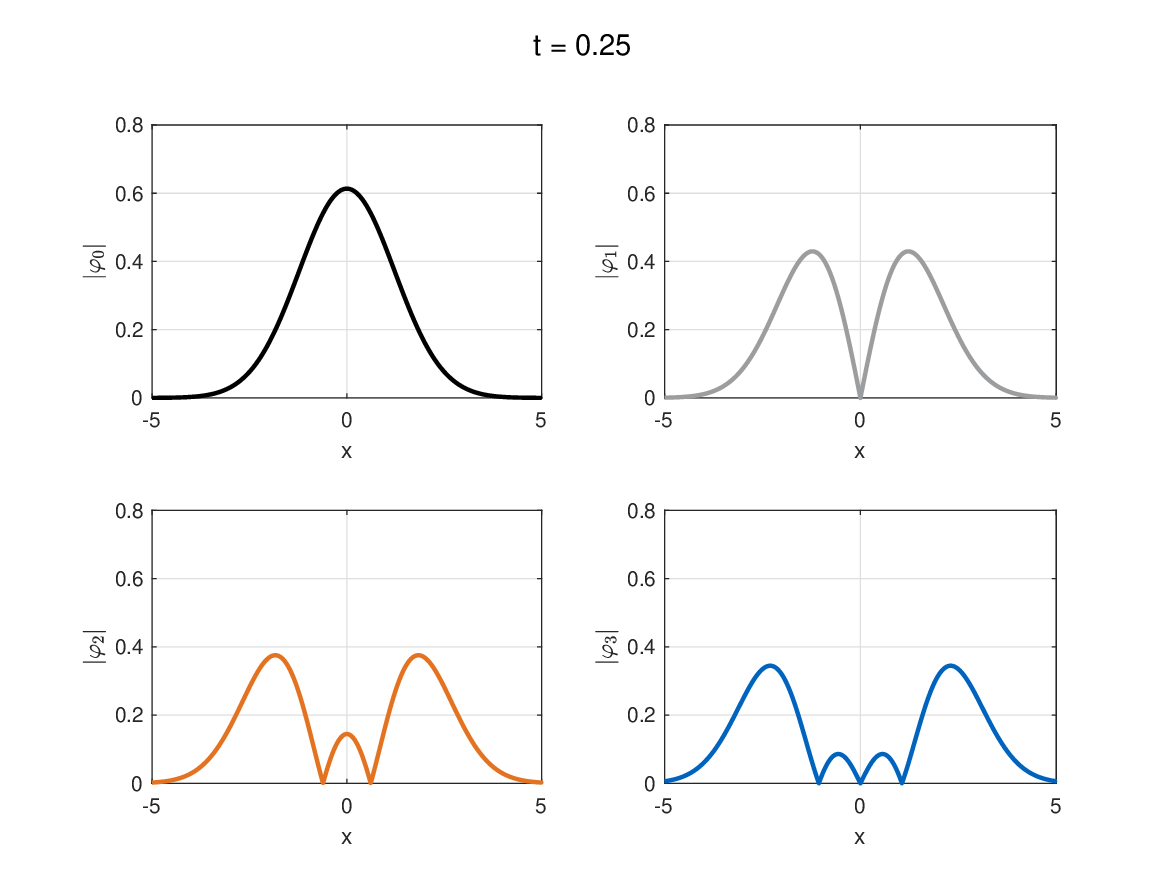}
\includegraphics[width=0.5\textwidth]{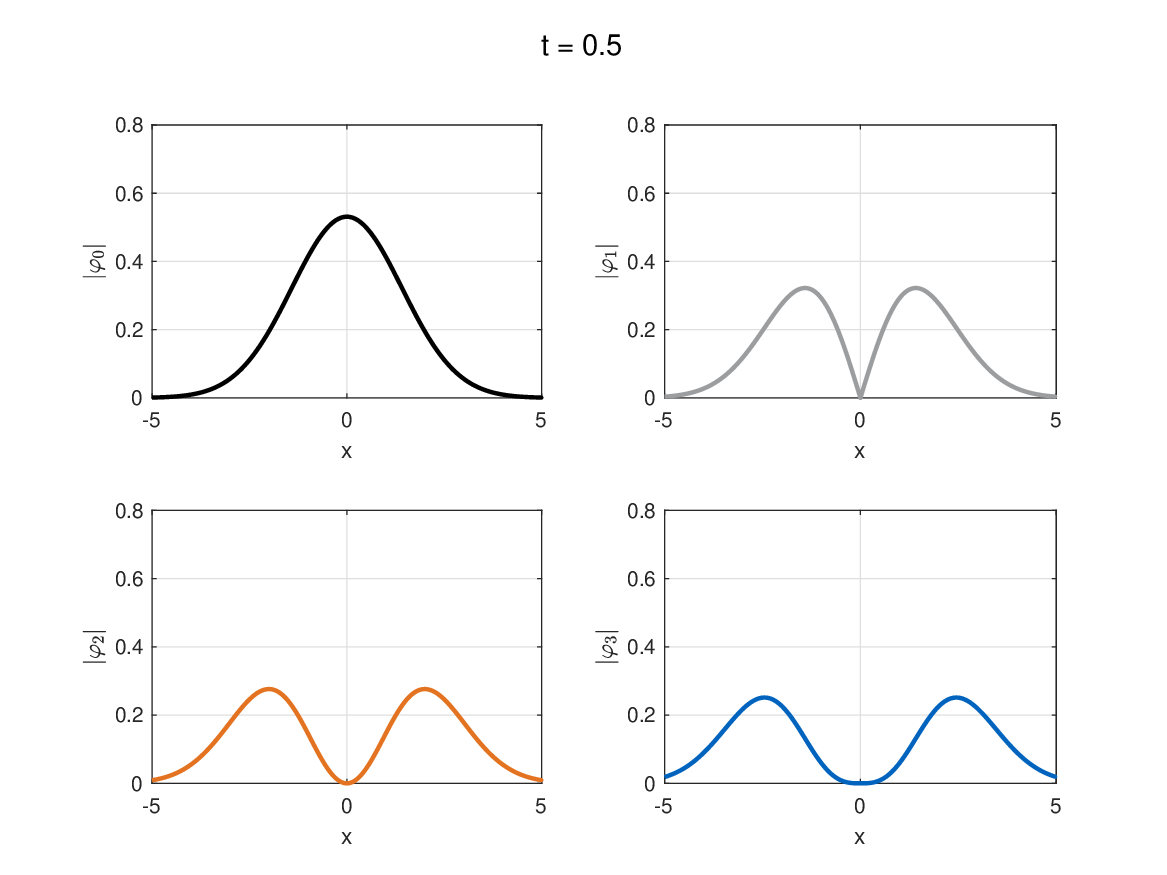}
\includegraphics[width=0.5\textwidth]{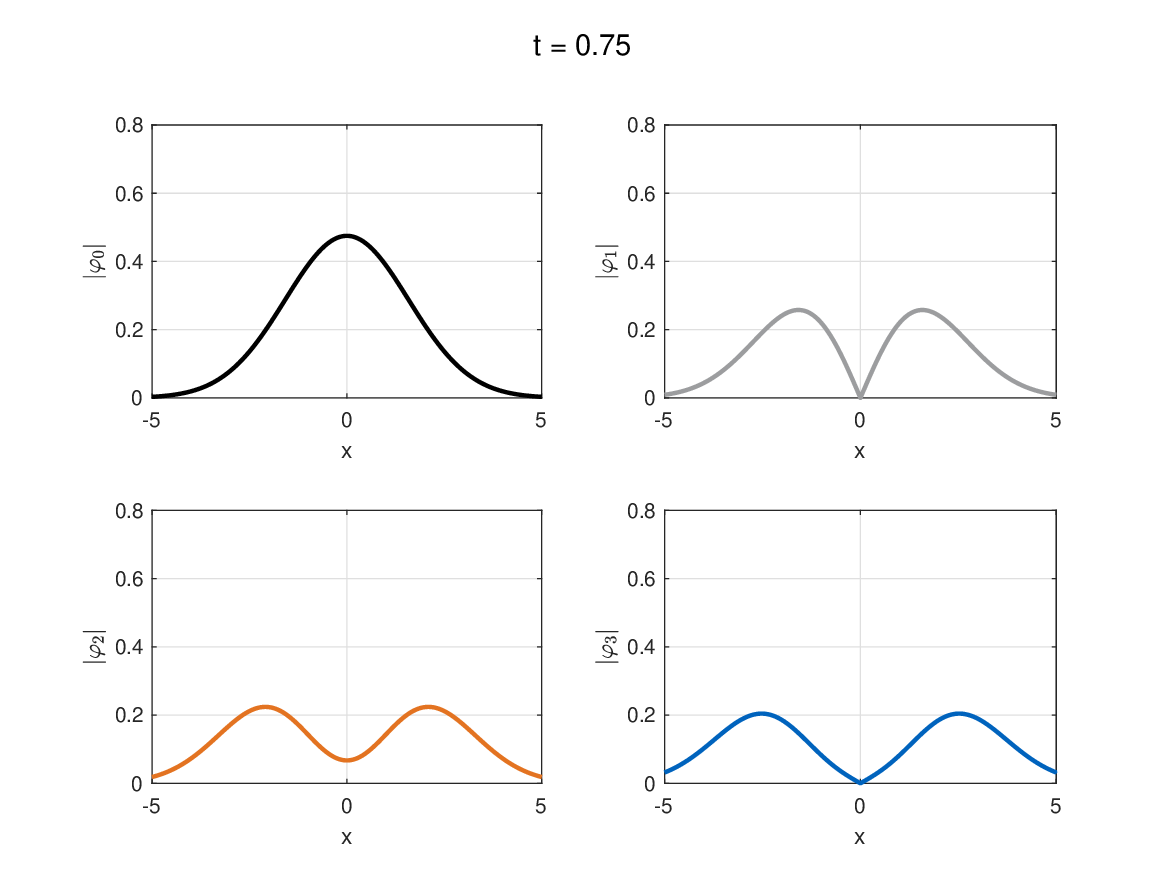}
\includegraphics[width=0.5\textwidth]{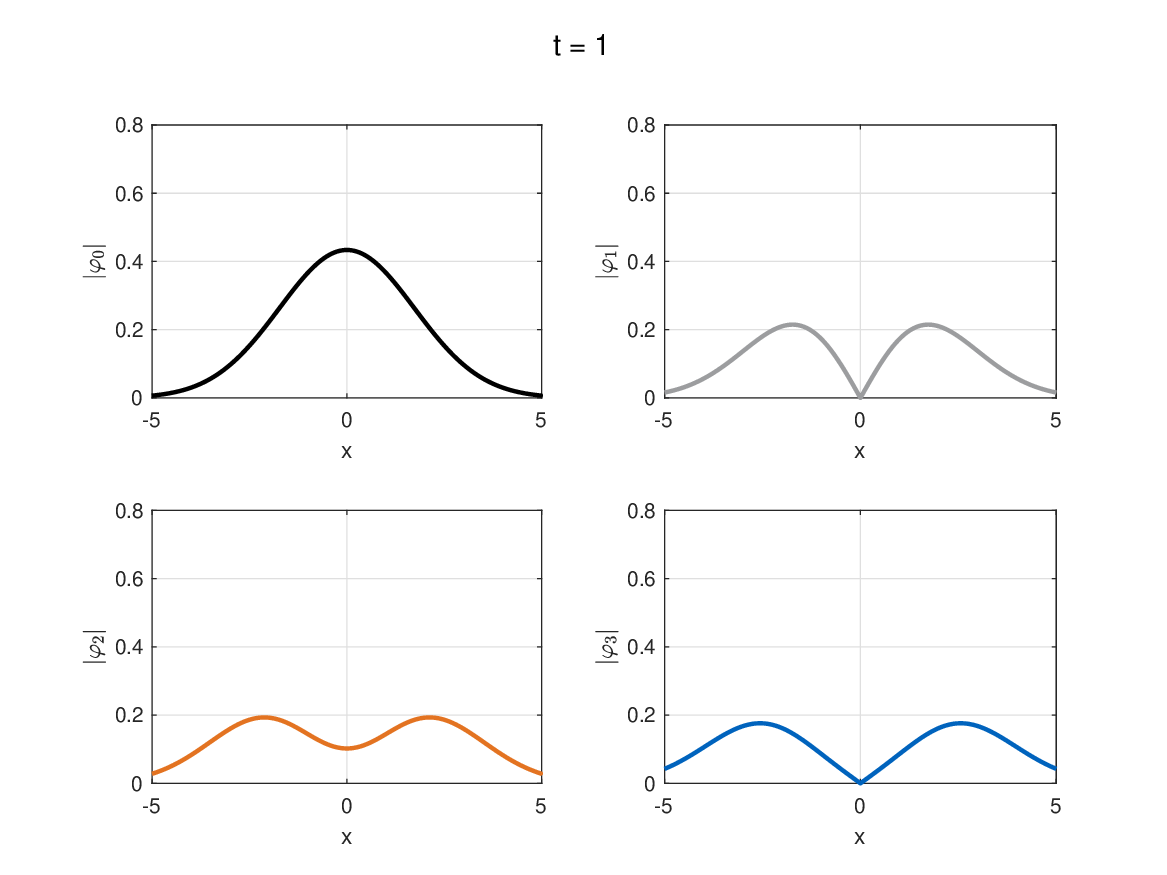}
\includegraphics[width=0.5\textwidth]{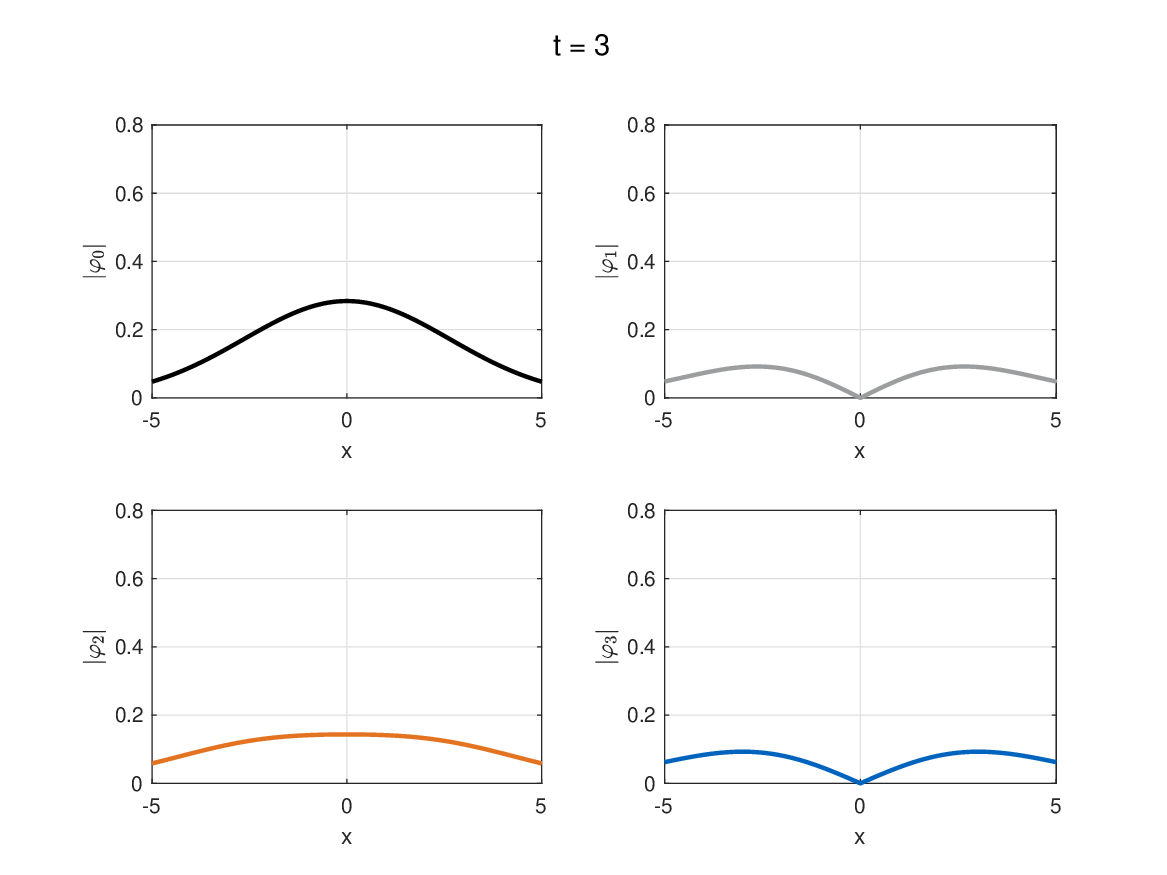}
\caption{Time evolution of  the state $|\varphi_k(t)|$ governed by the heat equation with initial data defined by $\l_0=(1,-\ui)$ at times $t=0$ (upper left), $t=0.25$ (upper right), $t=0.5$ (middle left), $t=0.75$ (middle right), $t=1$ (lower left) and $t=3$ (lower right). The color coding is black for $k=0$, grey for $k=1$, orange for $k=2$ and blue for $k=3$. We find the predicted behaviour, the roots of $\varphi_2$ and $\varphi_3$ are vanishing at $t=0.5$, the coherent state $\varphi_0$ is slowly damped due to the factor $e^{\beta_t}$.}\label{fig:normHeat}
\end{figure}
\subsection{Multivariate diffusion equation}\label{sec: Diff}
Our findings for the diffusion equation can easily be generalised to several dimensions. Let 
$$
\partial_t \rho  =  \alpha \Delta \rho  \quad \text{with} \quad H_D =-2\ui \alpha \begin{pmatrix} \Id_n & 0 \\ 0 & 0 \end{pmatrix}
$$
and $\alpha \in \C$. The flow emerges again as the matrix exponential
$$
S_t = \exp(t\Omega H_D)=\Id_n + t \, \Omega H_D = \begin{pmatrix} \Id_n & 0 \\ -2\ui \alpha t \, \Id_n & \Id_n\end{pmatrix}
$$
satisfying
$$
S^*_t \Omega S_t = \begin{pmatrix}4 \ui \Re(\alpha) t \, \Id_n & -\Id_n \\  \Id_n & 0\end{pmatrix} \ .
$$
Hence, our statements of Lemma~\ref{lem:diff1} on positive Lagrangian subspaces and the existence of the time evolution can immediately be lifted to several dimensions as follows:
\begin{lem}[Positive Lagrangian subspace]
Let $L_0 = \spann\{Z_0\}$ be a positive Lagrangian subspace spanned by a normalised Lagrangian frame $Z_0 = (P_0; Q_0) \in\C^{2n \times n}$. Then, the normalisation is given by
$$
N^{-2}_t = \Id_n + 2 \Re(\alpha) t \, P^*_0P_0
$$
and $L_t =S_t L_0$ is positive for all $t\in\R$ if $\Re(\alpha) \geq 0$.  Otherwise, denote by $\lambda_{\mathrm{max}}$ the eigenvalue of $P_0$ with the largest absolute value. Then, $L_t$ is positive for all $t\in[0,T[$ with $T=(-2\Re(\alpha) |\lambda_{\mathrm{max}}|^2)^{-1}$.
\end{lem}
For our futher investigations we consider two dimensions and choose as an anisotropic initial value
$$
P_0 = \tfrac{1}{2\sqrt{2}} \begin{pmatrix} -1+\ui & 1+\ui \\ 1+\ui & -1+\ui \end{pmatrix}, \quad Q_0 = \tfrac{1}{\sqrt{2}} \begin{pmatrix} 1+\ui & 1-\ui \\ 1-\ui & 1+\ui \end{pmatrix} \ .
$$
The corresponding wavepacket $\varphi_k(Z_0)$, $k \in \N^2$ and $Z_0 = (P_0;Q_0)$, is determined by the width matrix of the coherent state $P_0Q^{-1}_0 = \tfrac{\ui}{2} \Id_2$ and the recursion matrix of the polynomial prefactor 
$$
M_0 = Q^{-1}_0 \bar Q_0 = \begin{pmatrix} 0 & 1 \\ 1 & 0 \end{pmatrix} \  .
$$
Since $M_0$ is not a diagonal matrix the wavepackets are not simple tensor products of one-dimensional Hermite functions, but can be expressed by means of the Laguerre polynomials, see Appendix ~\ref{app:2d}.
For the time evolution we can infer from $P^*_0P_0 = \tfrac12 \Id_2$,
$$
N_t =  \left(\Id_2 + 2 \Re(\alpha) t \, P^*_0P_0 \right)^{-1/2}=(1+ \Re(\alpha) t)^{-1/2} \, \Id_2 \ ,
$$
and $e^{\beta_t} = (\det N_t)^{1/2} =  (1+ \Re(\alpha) t)^{-n/4}$. This information fully describes the propagated coherent state for $t \geq 0$ if $\Re(\alpha)\geq 0$ or $t \in [0;-\frac{1}{\Re(\alpha)}[$ otherwise.

For the evolution of the excited states we moreover need to determine the recursion matrix $\widetilde M_t$ defined in Corollary \ref{cor:pol}. We find
$$
\widetilde M_t = N^2_t \left(-\Re(\alpha) t + \frac{1+\bar \alpha t}{1+\alpha t}\right) \, \begin{pmatrix} 0 & 1 \\ 1 & 0 \end{pmatrix} \ ,
$$
so that the circular structure of the wavepackets discussed in Appendix \ref{app:2d} is preserved.

Depending on the sign of the real part of $\alpha$ we can distinguish two different cases for the diffusion. Our results thereby agree with the basic mathematical theory of diffusion: the diffusion rate $\Re(\alpha)$ is proportional to the gradient of the concentration that is modelled. In more detail, diffusion occurs in the opposite direction of the increasing concentration, see \cite[\S 1.2]{C75}. Hence, $\Re(\alpha)>0$ corresponds to a diffusion to the outside and a spreading coherent state. In this case our model is well-defined for all times $t \geq 0$ and the norm is slowly decaying. We illustrate this case by means of the classical heat equation, $\alpha = 1$. Then, 
$$
N_t = (1+t)^{-1/2} \, \Id_2 \ , \quad \widetilde M_t = \frac{1-t}{1+t} \begin{pmatrix} 0 & 1 \\ 1 & 0 \end{pmatrix} \ ,
$$ 
the wavepackets are damped. This behaviour is displayed in Figure \ref{fig:2Heat}.
\begin{figure}[h]
\includegraphics[width=0.5\textwidth]{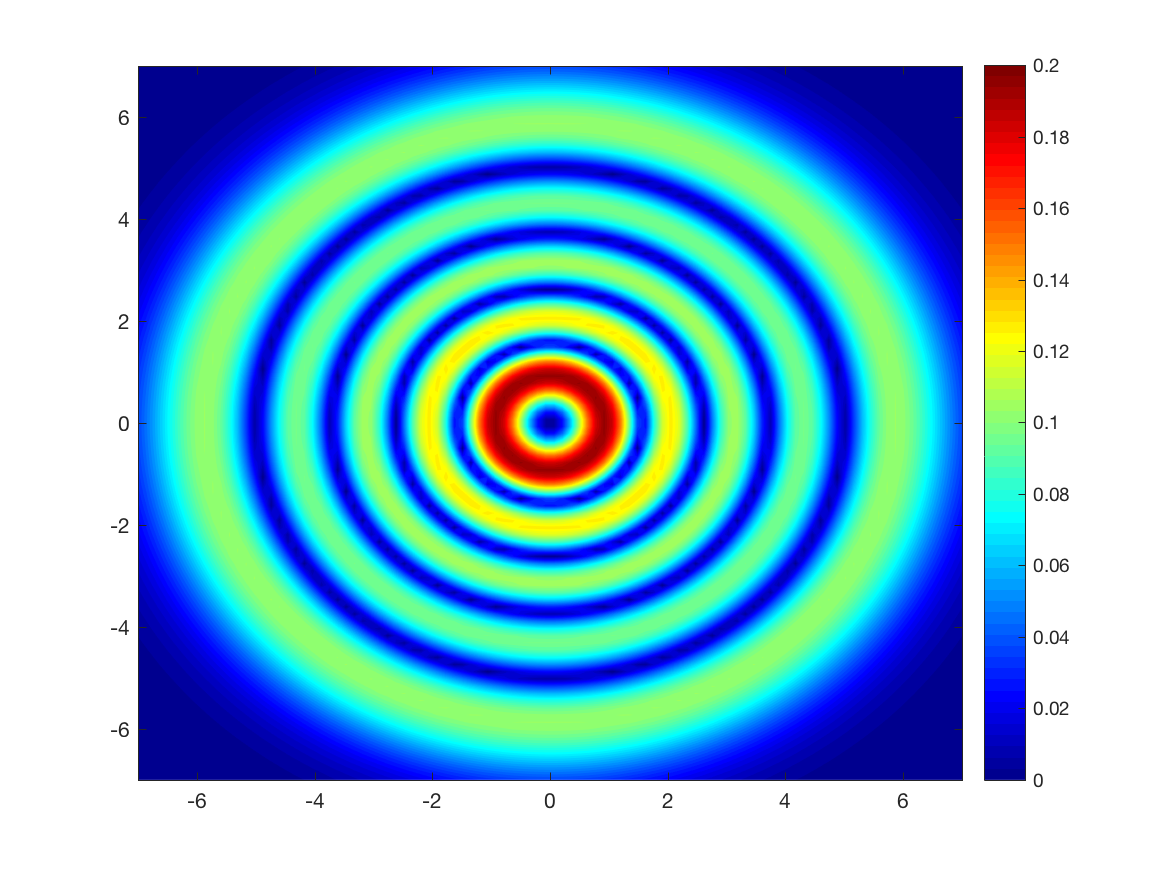}
\includegraphics[width=0.5\textwidth]{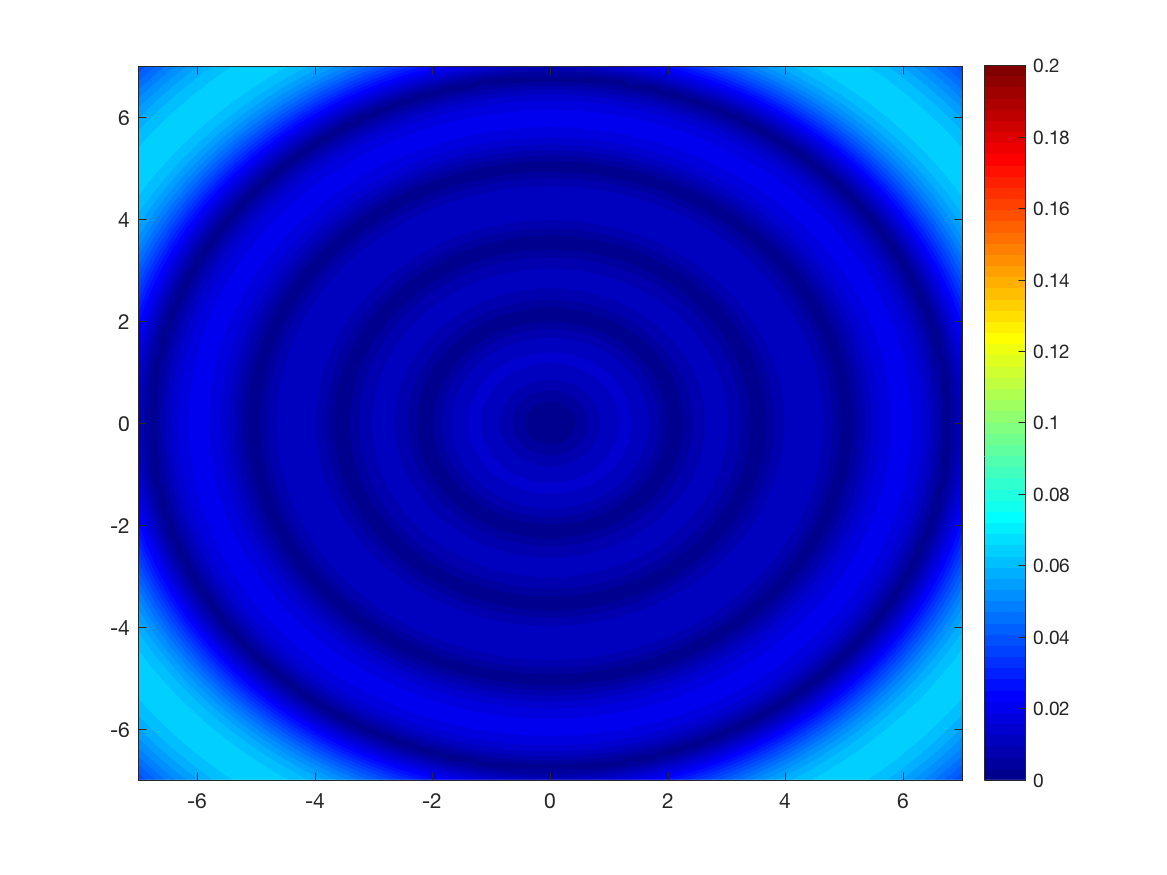}
\includegraphics[width=0.5\textwidth]{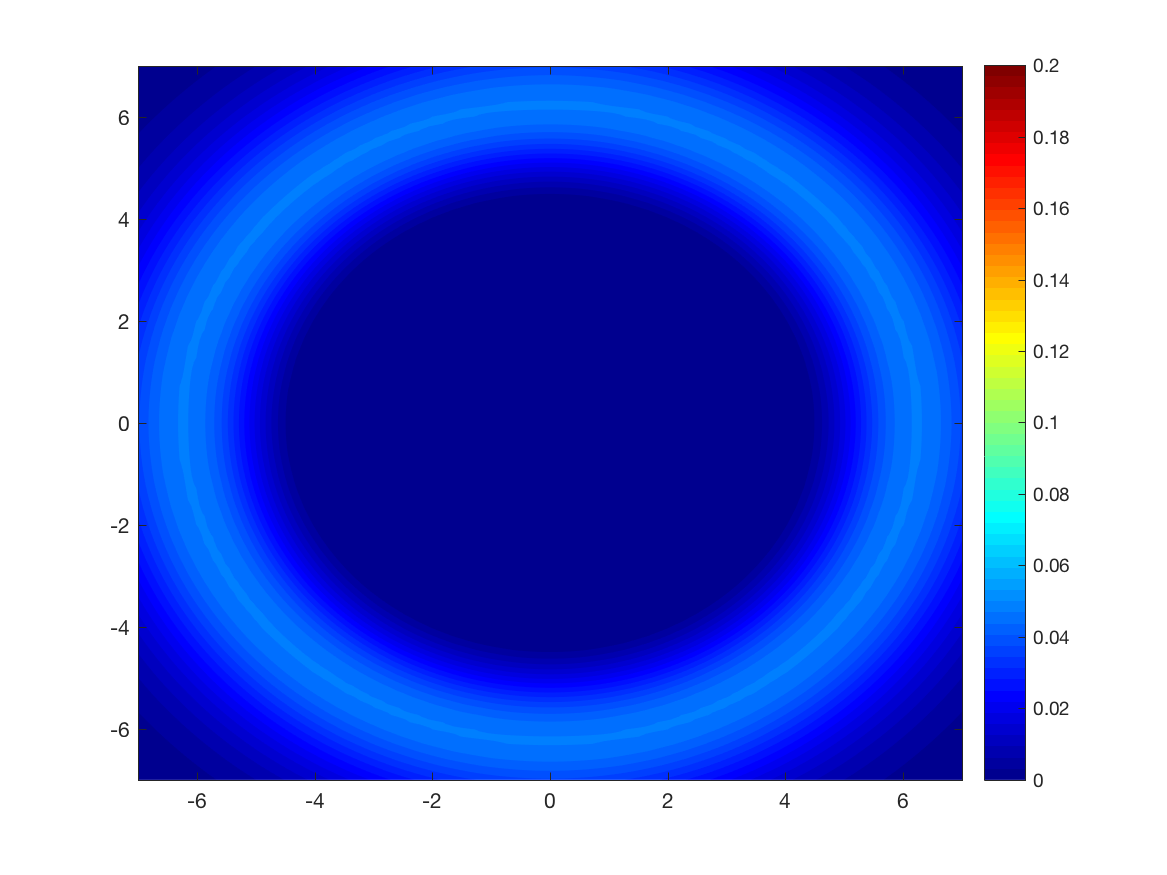}
\includegraphics[width=0.5\textwidth]{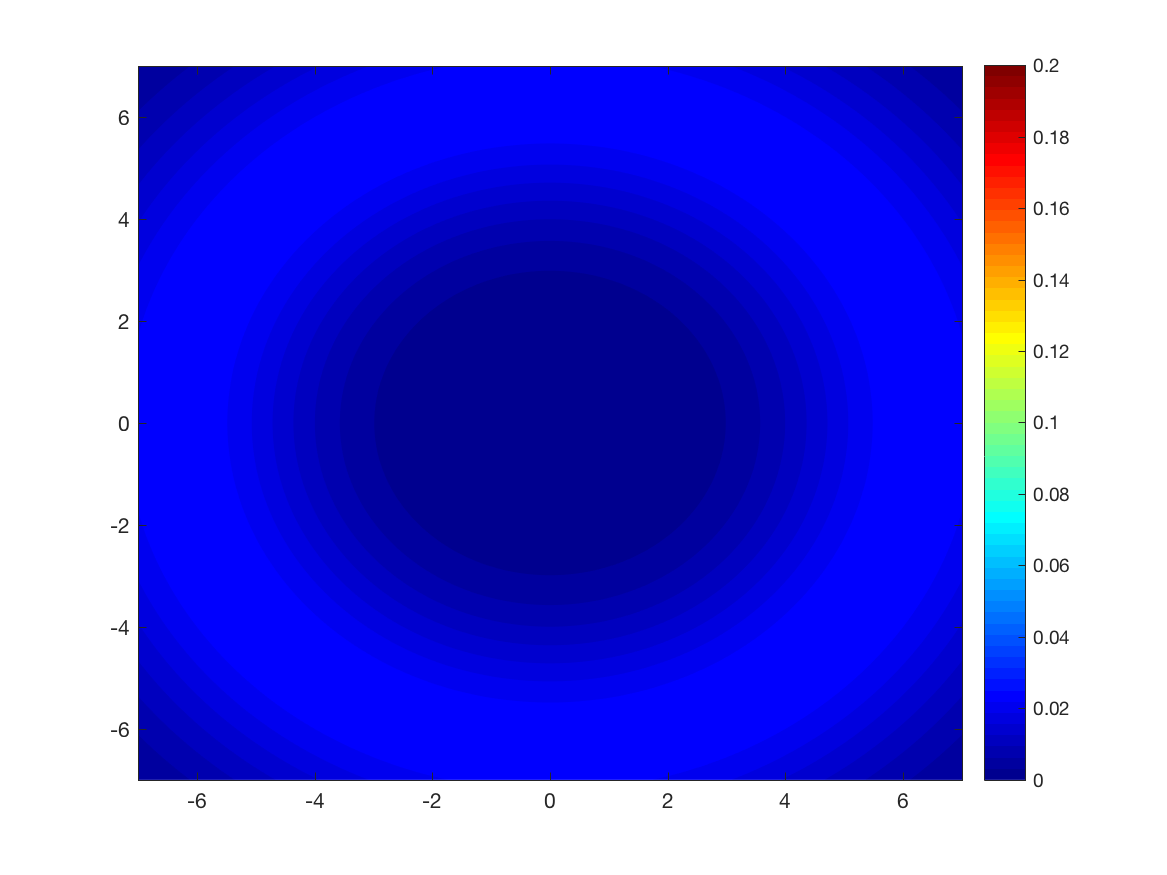}
\caption{Time evolution of  the state $|\varphi_k(t)|$ governed by the heat equation for $k=(4,6)$ at times $t=0$ (upper left), $t=0.25$ (upper right), $t=0.5$ (lower left) and $t=1.5$ (lower right). As $\Re(\alpha)>0$ the wavepacket is damped. Although the circular structure stays unaltered, real roots are vanishing due to the change of sign of the entries of $\widetilde M_t$.}\label{fig:2Heat}
\end{figure} 
On the other side $\Re(\alpha)<0$ models a diffusion to the origin. In our model the norm is increasing and the propagation collapses at $T= -\Re(\alpha)^{-1}$. We examine this case for $\alpha= -1+\tfrac{\ui}{4}$. Then,
$$
N_t = (1-t)^{-1/2} \, \Id_2 \quad \text{and} \quad \widetilde M_t = \frac{1}{1-t} \left(t + \frac{1-t-\tfrac{\ui}{4}t}{1-t+\tfrac{\ui}{4}t} \right) \begin{pmatrix} 0 & 1 \\ 1 & 0 \end{pmatrix} \ , 
$$ 
what corresponds to an increasing norm of the wavepackets. By taking $\Im(\alpha) \neq 0$ we additionally include dynamics induced by the Schr\"odinger part of the equation. Consequentially, the wavepackets tend to the origin at the beginning, but are then broadened again. Figure \ref{fig:2Heatb} shows the propagation for this setting.
\begin{figure}[h]
\includegraphics[width=0.5\textwidth]{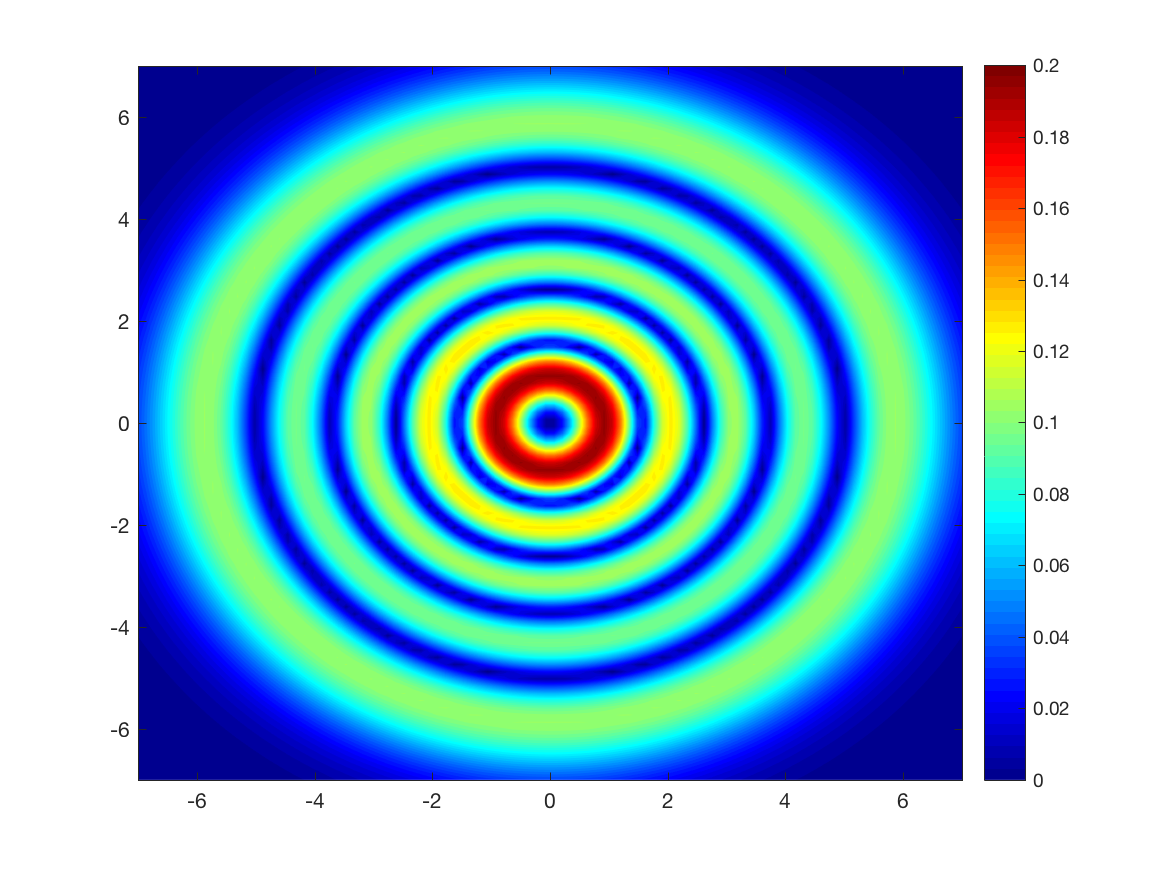}
\includegraphics[width=0.5\textwidth]{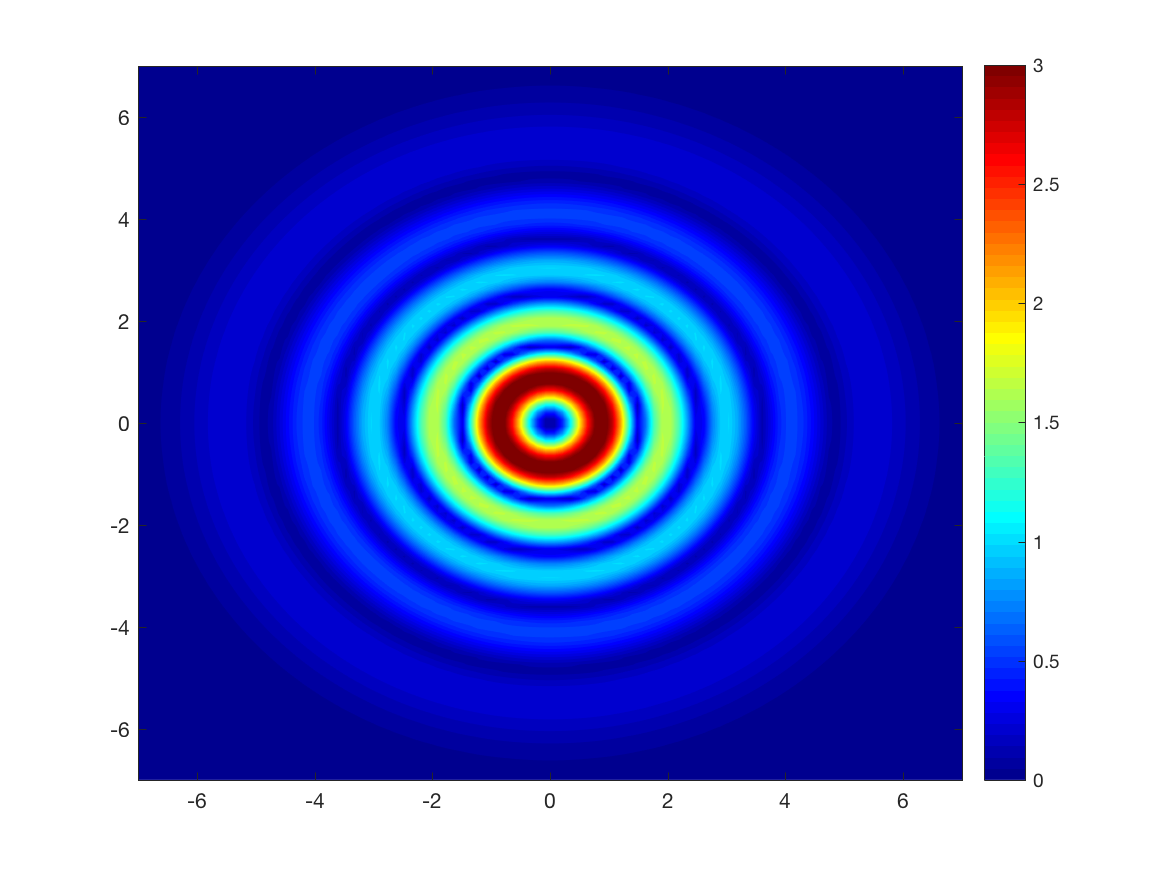}
\includegraphics[width=0.5\textwidth]{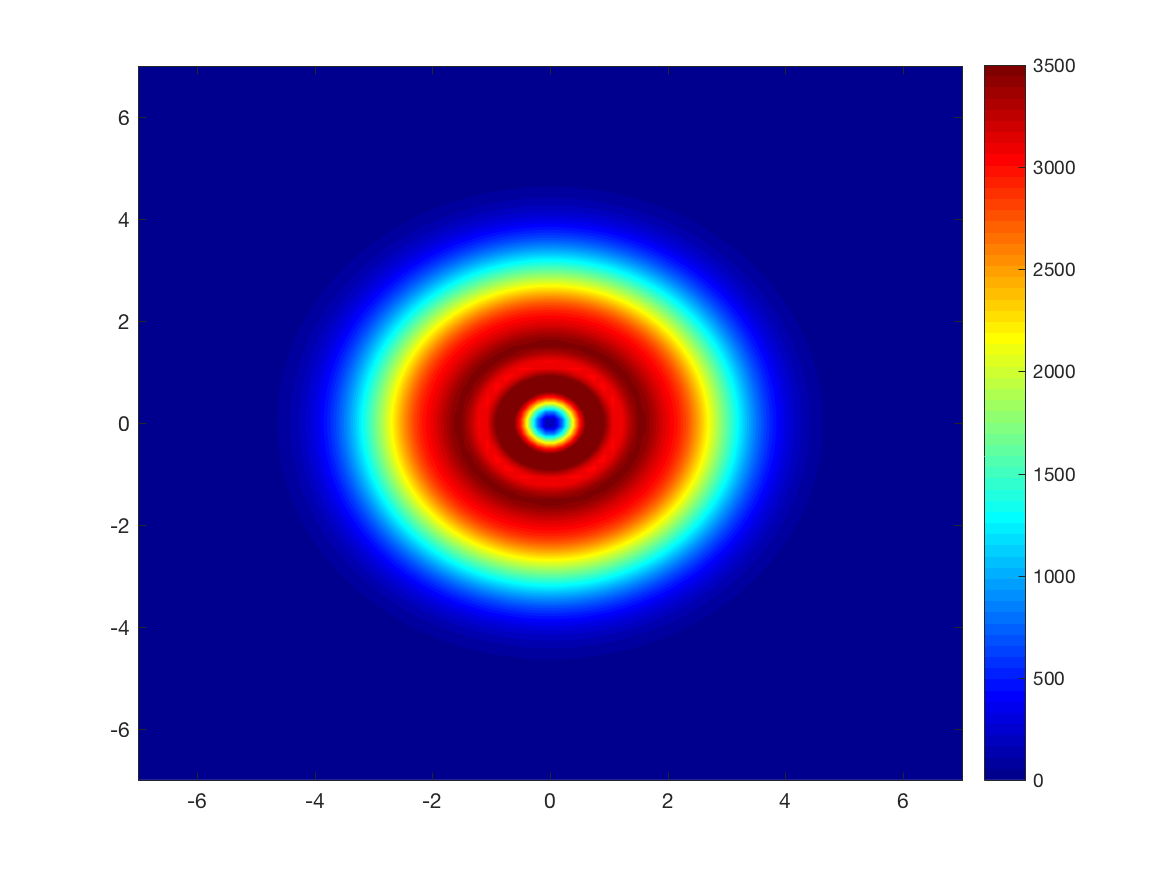}
\includegraphics[width=0.5\textwidth]{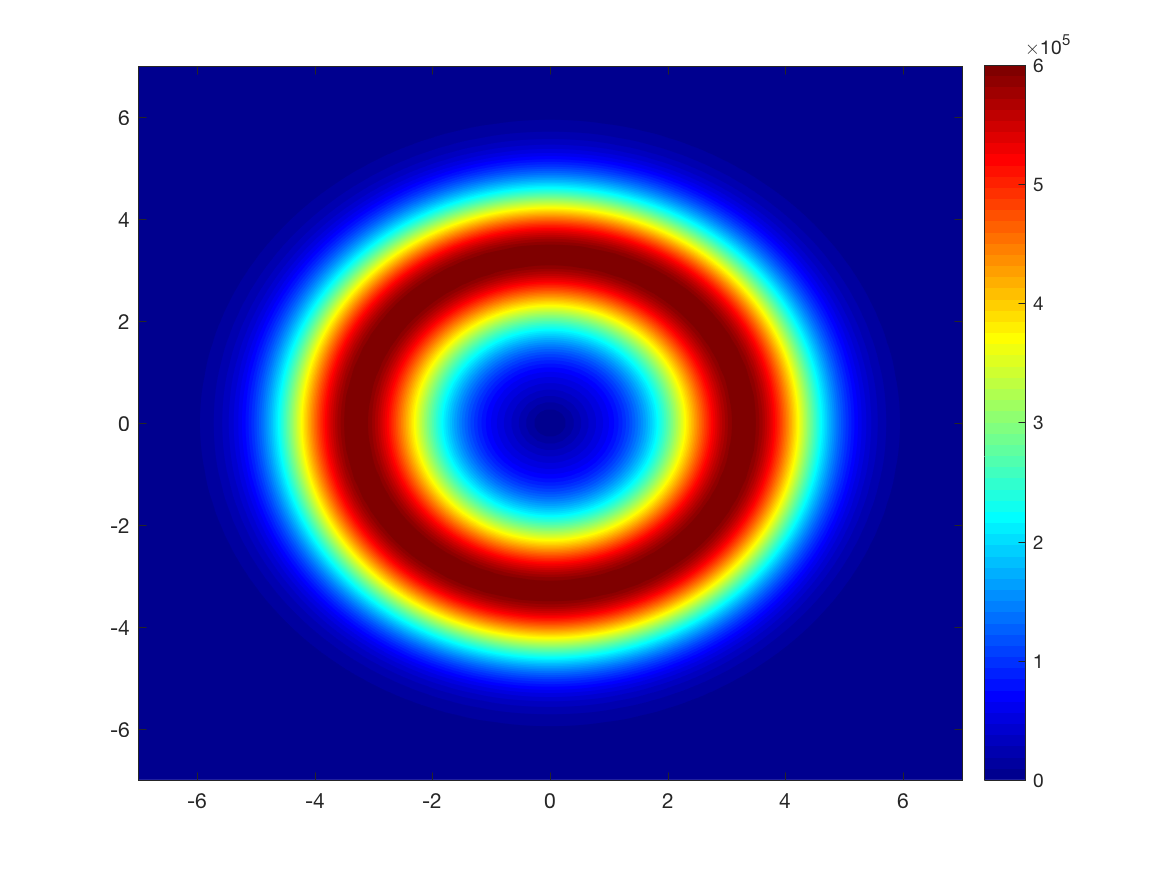}
\caption{Time evolution of  the state $|\varphi_k(t)|$, $k=(4,6)$, for $\alpha= -1+\tfrac{\ui}{4}$ at times $t=0$ (upper left), $t=0.25$ (upper right), $t=0.75$ (lower left) and $t=0.9$ (lower right). The norm of the wavepackets is significantly increasing over time. The circles first tend to the origin as $\Re(\alpha)<0$, but then the imaginary part of $\alpha$ encounters and broadens the evolution.}\label{fig:2Heatb}
\end{figure}

\appendix
\section{Weyl calculus}\label{app:Weyl}

Let us recall a few standard results about products and Weyl quantisation, see \cite[Chapter~2]{CR12} for background. We consider smooth phase space functions $a,b$ so that 
$$
\Op[a]\psi(x) = (2\pi\veps)^{-n} \int_{\R^{2n}} a(\xi,\tfrac{1}{2}(x+y)) \ue^{\frac{\ui}{\veps}\xi\cdot(x-y)}\psi(y) d\xi dy
$$
together with the compositions $\Op[a] \Op[b]$ and 
$\Op[b] \Op[a]$ are well-defined linear operators on dense subsets of $L^2(\R^n)$. The symbol of the operator product is the so-called {\em Moyal product} of $a$ and $b$,
$$
\Op[a]\Op[b] = \Op[a\sharp b] \ .
$$
If one of the two symbols $a$ or $b$ is a polynomial of degree $\le 2$, then
$$
a\sharp b = ab + \frac{\ui\veps}{2}\nabla a\cdot\Omega\nabla b - \frac{\veps^2}{8} \tr\left(D^2a \Omega D^2b \Omega^T\right) \ ,
$$
where $\nabla = \nabla_{p,q}$ and 
$$
\Omega = \begin{pmatrix}0 & -\Id_n\\ \Id_n & 0\end{pmatrix}\in\R^{2n\times 2n} \ .
$$
Consequently, the commutator can be written as
\begin{equation}
\label{eq:Moyal}
[\Op[a],\Op[b]]=\ui\veps \Op[\nabla a\cdot\Omega\nabla b] \ .
\end{equation}
In particular, the canonical commutation relations can be quickly verified as
$$
[\hat q_j,\hat p_k] = \ui \veps \nabla q_j \cdot\Omega \nabla p_k = 
\ui\veps \begin{pmatrix}0\\ e_j\end{pmatrix}\cdot\begin{pmatrix}0\\ e_k\end{pmatrix} = \ui\veps \delta_{jk} \ .
$$
Another application of the product rule yields that the Weyl quantisation of a symmetric quadratic form equals the quadratic form in $\hat z$.

\begin{lem}[Quadratic symbol]\label{lem:quad}
We consider 
$$
H = \begin{pmatrix} H_{pp} & H_{pq}\\ H_{qp} & H_{qq}\end{pmatrix}\in\C^{2n\times 2n} \ .
$$ 
Then, $\Op[z\cdot H z] = \hat z \cdot H\hat z + \frac{\ui\veps}{2}\tr(H_{qp}-H_{pq})$.  In particular, 
$$
\Op[z\cdot H z] = \hat z \cdot H\hat z \ ,\quad\text{if}\quad H=H^T \ .
$$
\end{lem}

\begin{proof}
We compute
$$
\hat z\cdot H\hat z = \Op[z\cdot Hz] + \frac{\ui\veps}{2} \sum_{j,k=1}^{2n} \nabla z_j\cdot\Omega H_{jk}\nabla z_k =  
\Op[z\cdot Hz] + \frac{\ui\veps}{2}\tr(-H_{pq} + H_{qp}) \ ,
$$
since 
$\sum_{j,k=1}^{2n} H_{jk} (e_j\cdot\Omega e_k) = -H_{1n}-\cdots-H_{n,2n} + H_{n1} + \cdots + H_{2n,n}
= \tr(-H_{pq} + H_{qp})$ \ .
\end{proof}

\section{Dynamics of the metric and the complex structure}\label{app:ric}

We provide a Lagrangian frame's proof for Theorem~\ref{thm:ric}, that states the Riccati equations for the symplectic metric and the complex structure of the positive Lagrangian $L_t = S_t L_0$, that is, 
\begin{eqnarray*}
\dot G_t &=& \Re H_t \Omega G_t - G_t\Omega \Re H_t - \Im H_t - G_t\Omega \Im H_t \Omega G \ ,\\
\dot J_t &=& \Omega\Re H_t J_t - J_t\Omega \Re H_t + \Omega \Im H_t + J_t\Omega \Im H_t J_t \ .
\end{eqnarray*}

\begin{proof}
We only work for $J_t$, since $G_t=\Omega J_t$. Let $Z_0\in\Fn(L_0)$ and consider an invertible matrix $N_t\in\C^{n\times n}$ so that 
$Z_t = S_tZ_0N_t\in\Fn(L_t)$. We then have $J_t=-\Re(Z_tZ_t^*)\Omega$. As in the proof of Propositon~\ref{prop:coh} we obtain
\begin{eqnarray*}
0 &=& \partial_t N_t^* N_t^{-*} +\tfrac{\ui}{2}N_t (S_tZ_0)^*(H_t - \bar H_t)(S_tZ_0)N_t + N_t^{-1}\partial_t N_t\\
&=& \partial_t N_t^* N_t^{-*} - Z_t^*\Im H_t Z_t + N_t^{-1}\partial_t N_t \ .
\end{eqnarray*}
Next we differentiate $Z_t$ so that
$$
\dot Z_t = \Omega H_t S_t Z_0 N_t + S_t Z_0 \partial_t N_t = \Omega H_t Z_t + Z_t N_t^{-1} \partial_t N_t \ .
$$
Therefore, 
\begin{eqnarray*}
\partial_t (Z_t Z_t^*) &=& 
\Omega H_t Z_t Z_t^* + Z_t N_t^{-1} \partial_t N_t Z_t^* + Z_t Z_t^* \bar H_t \Omega ^T + Z_t \partial_t N_t^* N_t^{-*} Z_t^*\\
&=& 
\Omega H_t Z_t Z_t^* + Z_t Z_t^* \bar H_t \Omega ^T + Z_t  Z_t^*\Im H_t Z_tZ_t^* \ .
\end{eqnarray*}
Since $\Im(Z_tZ_t^*) = -\Omega$, we then have
\begin{eqnarray*}
\partial_t\Re(Z_tZ_t^*) &=& \Omega \Re H_t \Re(Z_tZ_t^*) + \Omega\Im H_t\Omega - \Re(Z_tZ_t^*)\Re H_t \Omega\\
&&  + \Re(Z_tZ_t^*)\Im H_t\Re(Z_tZ_t^*)
\end{eqnarray*}
and the claimed equation
$
\dot J_t = \Omega \Re H_t J_t + \Omega \Im H_t - J_t \Omega\Re H_t + J_t\Omega\Im H_t J_t \ .
$
\end{proof}

\section{Multivariate polynomials}\label{sec:pol}

Analysing Hagedorn wave packets and their dynamics, we have encountered multivariate polynomials generated by the following type of recursion relation.

\begin{Def}[Polynomial recursion]\label{def:pol} 
Let $M\in \C^{n\times n}$ be symmetric and $c\in\C$. We define a set of multivariate polynomials $p_{\alpha}(x)$ by the recursion relation
\begin{equation}\label{eq:rec_app}
p_0(x) = c \ , \qquad \left( p_{\alpha+e_j}(x)\right)^{n}_{j=1} = x p_{\alpha}(x) - M \nabla p_{\alpha}(x)
\end{equation}
with $x\in \C^n$ and $\alpha\in \N_0^n$.
\end{Def}
Together with $c=1$ the matrix $M=0$ generates the monomials $p_\alpha(x)=x^\alpha$, while the identity matrix $M= \Id$ determines tensor products of simple Hermite polynomials. \\ If $Q\in\C^{n\times n}$ is the lower block of a normalised Lagrangian frame
$$
Z = \begin{pmatrix}P\\ Q\end{pmatrix} \in \C^{2n \times n} \ ,
$$
then the matrix $M = Q^{-1} \bar Q$ generates the polynomial prefactor $p^M_\alpha$ of the Hagedorn wave packets, that is, 
\begin{equation}\label{eq: PolFac}
\varphi_\alpha(Z;x) = \tfrac{1}{\sqrt{\alpha!}} p^M_\alpha\left(\sqrt{\tfrac{2}{\veps}}\,Q^{-1}x\right)\, \varphi_0(Z;x) \ , \qquad \alpha \in \N^n \ .
\end{equation}
\begin{lem}[Properties of the recursion matrix]
Let $Z=(P;Q) \in \C^{2n \times n}$ be a normalised Lagrangian frame. Then, $M=Q^{-1} \bar Q$ is unitary and symmetric.
\end{lem}
\begin{proof}
As argued in Corollary \ref{cor:pol} the symmetry of $M$ follows since $QQ^* = \left( \Im(PQ^{-1})\right)^{-1}$ is real symmetric,
$$
M-M^T = Q^{-1} \bar Q-Q^*Q^{-T} = Q^{-1} \left( \bar QQ^T - Q Q^* \right)Q^{-T} = 0 \, .
$$
The unitarity is a direct consequence,
$$
MM^* = M\bar M = Q^{-1} \bar Q \bar Q^{-1} Q = \Id_n \, .
$$
\end{proof}
All the polynomials sequences of Definition~\ref{def:pol} are multivariate versions of orthogonal polynomials determined by Favard's theorem.
In the univariate setting, a polynomial sequence $p_n(x)$, $n\in\N_0$, is called an {\em Appell sequences}, if $p_n'(x)=np_{n-1}(x)$ for all $n\ge1$. 
The Hermite polynomials are prominent examples and the only orthogonal Appell sequence. In several dimensions this property generalises to the following gradient formula, which is due to \cite[Lemma 6]{DiKeTr15}.

\begin{lem}[Gradient formula]\label{lem:gradient} Let $M\in\C^{n\times n}$ be symmetric and $c\in\C$. The polynomials defined by the recursion relation \eqref{eq:rec_app} satisfy
$$
\nabla_x p_{\alpha}(x)  =  \left( \alpha_j p_{\alpha-e_j}(x)\right)^{n}_{j=1}
$$
for all $\alpha\in\N_0^n$ and $x\in\C^n$.
\end{lem}

\begin{proof}
We argue by induction and assume that the gradient formula holds for a fixed $\alpha\in\N_0^n$. Differentiating the recursion relation, we get
\begin{eqnarray*}
\partial_k p_{\alpha+e_j} &=& \delta_{kj} p_\alpha + x_j\partial_k p_\alpha - e_j\cdot M\nabla\partial_k p_\alpha = 
\delta_{kj}p_\alpha + \alpha_k(x_j p_{\alpha-e_k} - e_j\cdot M\nabla p_{\alpha-e_k})\\
&=&
\delta_{kj}p_\alpha + \alpha_k p_{\alpha+e_j-e_k} = (\alpha+e_j)_k p_{\alpha+e_j-e_k} \ .
\end{eqnarray*}
\end{proof}
\section{Wavepackets in one and two dimensions}\label{app:2d}
The examples in Section \ref{sec:swanson} focus on Hagedorn wavepackets in one or two dimensions. To explain the varying forms we encounter we briefly discuss their roots. Let $Z=(P;Q) \in \C^{2n \times n}$ be a normalised Lagrangian frame and $M=Q^{-1} \bar Q$. The corresponding wavepackets then emerge as (\ref{eq: PolFac}).
Hence, the roots of the wavepackets can be directly deduced from the roots of the polynomial $p^M_\alpha$. 
\subsection{One dimension}\label{ssec: 1Droots}
For a positive Lagrangian subspace $L = \spann\{l\}$, $l \in \C \oplus \C$, the corresponding states are simply rescaled Hermite functions, i.e. with $l=(p,q)$ we find for all $\alpha \in \N$
$$
\varphi_\alpha(l;x) = \tfrac{m^{\alpha/2}}{\sqrt{\alpha!}} \, h_\alpha(\tfrac{2}{\sqrt{\veps}}\tfrac{x}{q\sqrt{m}}) \, \varphi_0(l;x) \, 
$$ 
where $m=q^{-1} \bar q \in \C$ and $h_k$ denotes the $\alpha$-th probabilistic Hermite polynomial defined by
$$
h_{-1}(x)=0, \qquad h_0(x)=1, \qquad h_{\alpha+1}(x) = x h_\alpha(x) - \alpha h_{\alpha-1}(x) \quad \forall \alpha \geq 0 \, .
$$
The roots of the wavepacket therefore depend on $q\sqrt{m}=q\left( \frac{\bar q}{q} \right)^{1/2}$. If this value is real, the wavepacket shows $\alpha$ distinct roots, otherwise the roots vanish, see Section \ref{sec:rootsDiff}. 
\subsection{Two dimensions}
In two dimensions, the wavepackets relate to the Hermite functions only in special cases. The following result is a special case of \cite[Corollary 9]{DiKeTr15}. We provide a proof here for a self-contained reading.
\begin{lem}
Let 
$$
M = \begin{pmatrix} m_1 & m_3 \\ m_3 & m_2 \end{pmatrix} \in \C^{2\times 2} \, .
$$
and consider polynomials $(p^M_\alpha)_{\alpha \in \N^2}$ generated by $M$ via the recursion relation $p^M_{\alpha}(x)=0$ if $\alpha \notin \N^2$,
$$
p_0^M(x)=1, \qquad \begin{pmatrix} p_{(\alpha_1+1,\alpha_2)} \\ p_{(\alpha_1,\alpha_2+1)} \end{pmatrix} (x)= x p_\alpha(x) - M \begin{pmatrix} \alpha_1 p_{(\alpha_1-1,\alpha_2)} \\ \alpha_2 p_{(\alpha_1,\alpha_2-1)} \end{pmatrix}(x) \quad \forall x \in \R^2 \, .
$$
If $m_3 =0$ the polynomials $(p^M_\alpha)_{\alpha \in \N^2}$ can be written as
$$
p^M_\alpha(x) = m^{\alpha_1/2}_1m^{\alpha_2/2}_2 \, h_{\alpha_1} \left(\tfrac{x_1}{\sqrt{m_1}}\right)h_{\alpha_2}\left(\tfrac{x_2}{\sqrt{m_2}}\right) \, .
$$
If $m_1 = m_2=0$ the polynomials appear as
$$
p^M_\alpha(x) = \begin{cases} (-m_3)^{\alpha_2} \alpha_2! x_1^{\alpha_1-\alpha_2} L^{(\alpha_1-\alpha_2)}_{\alpha_2}\left(\frac{x_1x_2}{m_3}\right) & \text{if} \ \alpha_1 \geq \alpha_2 \\ (-m_3)^{\alpha_1} \alpha_1! x_2^{\alpha_2-\alpha_1} L^{(\alpha_2-\alpha_1)}_{\alpha_1}\left(\frac{x_1x_2}{m_3}\right) & \text{if} \ \alpha_1 < \alpha_2 \end{cases} \ ,
$$
where $L^{(\gamma)}_\ell$ denotes the $\ell$-th associated or generalised Laguerre polynomial.
\end{lem}
\begin{proof}
Both cases follow by induction. For $m_3 =0$, we have
\begin{align*}
p_{(\alpha_1+1,\alpha_2)}(x) & = x_1 p_{(\alpha_1,\alpha_2)}(x)- m_1 \alpha_1 p_{(\alpha_1-1,\alpha_2)}(x) \\
&= m^{\alpha_2/2}_2h_{\alpha_2}\left(\tfrac{x_2}{\sqrt{m_2}}\right) \left( m^{\alpha_1/2}_1 x_1h_{\alpha_1} \left(\tfrac{x_1}{\sqrt{m_1}}\right) -m^{(\alpha_1+1)/2}_1 \alpha_1h_{\alpha_1-1} \left(\tfrac{x_1}{\sqrt{m_1}}\right) \right) \\
& = m^{(\alpha_1+1)/2}_1m^{\alpha_2/2}_2h_{\alpha_2}\left(\tfrac{x_2}{\sqrt{m_2}}\right) \left( \tfrac{x_1}{\sqrt{m_1}}h_{\alpha_1} \left(\tfrac{x_1}{\sqrt{m_1}}\right) - \alpha_1h_{\alpha_1-1} \left(\tfrac{x_1}{\sqrt{m_1}}\right) \right) \\
&=m^{(\alpha_1+1)/2}_1m^{\alpha_2/2}_2h_{\alpha_1+1} \left(\tfrac{x_1}{\sqrt{m_1}}\right)h_{\alpha_2}\left(\tfrac{x_2}{\sqrt{m_2}}\right)  \, .
\end{align*}
The claim for $p_{(\alpha_1,\alpha_2+1)}$ can be proven similarly. For the generalised Laguerre polynomials we can use
$$
L^{(\gamma)}_\ell = L^{(\gamma+1)}_\ell-L^{(\gamma+1)}_{\ell-1} 
$$
for all $\gamma, \ell \in \N$. Then we find for $m_1=m_2 =0$ if $\alpha_1 \geq \alpha_2$,
\begin{align*}
p_{(\alpha_1+1,\alpha_2)}(x) & = x_1 p_{(\alpha_1,\alpha_2)}(x)- m_3 \alpha_2 p_{(\alpha_1,\alpha_2-1)}(x) \\
&=   (-m_3)^{\alpha_2} \alpha_2! x_1^{\alpha_1+1-\alpha_2} \left( L^{(\alpha_1-\alpha_2)}_{\alpha_2}\left(\tfrac{x_1x_2}{m_3}\right) + L^{(\alpha_1-\alpha_2+1)}_{\alpha_2-1}\left(\tfrac{x_1x_2}{m_3}\right) \right)\\
& = (-m_3)^{\alpha_2} \alpha_2! x_1^{\alpha_1+1-\alpha_2}  L^{(\alpha_1+1-\alpha_2)}_{\alpha_2}\left(\tfrac{x_1x_2}{m_3}\right) 
\end{align*}
and
\begin{align*}
p_{(\alpha_1,\alpha_2+1)}(x) & = x_2 p_{(\alpha_1,\alpha_2)}(x)- m_3 k_1 p_{(\alpha_1-1,\alpha_2)}(x) \\
& = (-m_3)^{\alpha_2+1} \alpha_2! x_1^{\alpha_1-\alpha_2-1} \left( - \tfrac{x_1x_2}{m_3} L^{(\alpha_1-\alpha_2)}_{\alpha_2}\left(\tfrac{x_1x_2}{m_3}\right) +\alpha_1L^{(\alpha_1-\alpha_2-1)}_{\alpha_2}\left(\frac{x_1x_2}{m_3}\right)  \right)\, . 
\end{align*}
With 
\begin{align*}
(\ell+1) L^{(\gamma)}_{\ell+1}(y) & = (2\ell+1+\gamma-x) L^{(\gamma)}_{\ell}(y) - (\ell+\gamma)L^{(\gamma)}_{\ell-1}(y) \, , \\
\ell L^{(\gamma)}_{\ell}(y) & = (\ell+\gamma) L^{(\gamma)}_{\ell-1}(y)-y L^{(\gamma+1)}_{\ell-1}(y)
\end{align*}
it moreover holds
\begin{align*}
(\alpha_2+1) L^{(\alpha_1-\alpha_2-1)}_{\alpha_2+1}(y) & = (\alpha_2+\alpha_1-y) L^{(\alpha_1-\alpha_2-1)}_{\alpha_2}(y) - (\alpha_1-1) L^{(\alpha_1-k_2-1)}_{\alpha_2-1}(y) \\
& = \alpha_1 L^{(\alpha_1-\alpha_2-1)}_{\alpha_2}(y)-y( L^{(\alpha_1-\alpha_2-1)}_{\alpha_2} + L^{(\alpha_1-\alpha_2)}_{\alpha_2-1}) \\
 & = -yL^{(\alpha_1-\alpha_2)}_{\alpha_2}(y) +\alpha_1L^{(\alpha_1-\alpha_2-1)}_{\alpha_2}(y) \, .
\end{align*}
Inserting $y = \tfrac{x_1x_2}{m_3}$ finishes the proof.
\end{proof}
Let
$$
Q=\begin{pmatrix} q_{11} & q_{12} \\q_{21} & q_{22} \end{pmatrix} \in \C^{2 \times 2} \, .
$$
If $m_3=0$, the factorisation $M=Q^{-1} \bar Q$ implies $m_1 = \bar m_2$ with $|m_1|=1$ and the Hagedorn wavepacket $\varphi_\alpha(Z)$ possesses at most $|\alpha|$-many real roots given by
\begin{align*}
q_{22}x_1 - q_{12}x_2 & = c_1 \eta_{1,j} \, , \qquad j = 1, \ldots, \alpha_1 \, ,\\
-q_{21}x_1 + q_{11}x_2 & = c_2 \eta_{2,j} \, , \qquad j = 1, \ldots, \alpha_2 \, ,
\end{align*}
where $\eta_{i,j}$ denotes the $j$-th root of $h_{\alpha_i}$ and $c_i = \sqrt{\veps m_i} \det(Q)$, $i=1,2$. Consequently, the roots form a lattice.

A matrix $Q$ that produces $m_1 = m_2 = 0$ is of the form $q_{12} = \bar m_3 q_{11}$, $q_{22} = \bar m_3 q_{21}$. The unitarity of $M$ again yields $|m_3|=1$. The roots of $\varphi_\alpha(Z)$ therefore satisfy for $\alpha_1 \geq \alpha_2$
$$
(\bar q_{21}x_1 - \bar q_{11}x_2)^j = 0 \, , \qquad j = 1, \ldots, \alpha_1-\alpha_2 \, ,
$$
and 
\begin{align*}
(\bar q_{21}x_1 - \bar q_{11}x_2)(q_{21}x_1 - q_{11}x_2) & = c_3 \theta_j, \, , \qquad j = 1, \ldots, \alpha_2 \, , \\
|q_{21}|^2 x_1 - |q_{11}|^2 x_2 - 2 \Re(q_{11} \bar q_{12}) x_1x_2 & = c_3 \theta_j, \, , \qquad j = 1, \ldots, \alpha_2 \, , \\
\end{align*}
where $\theta_j$, $j=1, \ldots, \alpha_2$ denotes the $j$-th root of the generalised Laguerre polynomial $L^{(\alpha_1-\alpha_2)}_{\alpha_2}$ and $c_3 = 4 \veps \Im(q_{11} \bar q_{12})$. So, the roots of the wavepackets consist of at most $\alpha_1-\alpha_2$-many lines through the origin and $\alpha_2$-many ellipses whose distances are proportional to the roots of $L^{(\alpha_1-\alpha_2)}_{\alpha_2}$, see Figure \ref{fig:2Heat} and \ref{fig:2Heatb}. 
\section*{Acknowledgements}
This research was supported by the German Research Foundation (DFG), Collaborative Research Center SFB-TRR 109.
\small
\bibliographystyle{amsalpha}

\end{document}